\documentclass[%
 reprint,
 amsmath,amssymb,mathtools,
 aps,
 pra,
 superscriptaddress
]{revtex4-2}
\usepackage{multirow}

\usepackage{graphicx}
\usepackage{dcolumn}
\newcolumntype{d}[1]{D{.}{.}{#1}}
\usepackage{bm}
\usepackage{braket}
\usepackage[english]{babel}
\usepackage{xcolor}
\makeatletter
\@ifundefined{l@en}{\let\l@en\l@english}{}
\@ifundefined{captionsen}{}{}
\@ifundefined{dateen}{}{}
\@ifundefined{extrasen}{}{}
\@ifundefined{noextrasen}{}{}
\makeatother
\usepackage{hyperref}
\hypersetup{
    colorlinks=true,
    linkcolor=blue,
    filecolor=blue,
    urlcolor=blue,
    citecolor=blue,
    pdfpagemode=FullScreen,
}
\usepackage[caption=false]{subfig}
\newcommand{\norm}[1]{\left\lVert#1\right\rVert}
\newenvironment{proof}{\noindent\textit{Proof.---}}{\hfill$\square$\medskip}

\newcounter{thm}[section]
\renewcommand{\thethm}{\thesection.\arabic{thm}}
\newenvironment{theorem}[1][]%
  {\refstepcounter{thm}\medskip\noindent\textbf{Theorem~\thethm\if\relax\detokenize{#1}\relax\else\space(#1)\fi.}\enspace\itshape}%
  {\medskip}
\newenvironment{lemma}[1][]%
  {\refstepcounter{thm}\medskip\noindent\textbf{Lemma~\thethm\if\relax\detokenize{#1}\relax\else\space(#1)\fi.}\enspace\itshape}%
  {\medskip}
\newenvironment{proposition}[1][]%
  {\refstepcounter{thm}\medskip\noindent\textbf{Proposition~\thethm\if\relax\detokenize{#1}\relax\else\space(#1)\fi.}\enspace\itshape}%
  {\medskip}
\newenvironment{corollary}[1][]%
  {\refstepcounter{thm}\medskip\noindent\textbf{Corollary~\thethm\if\relax\detokenize{#1}\relax\else\space(#1)\fi.}\enspace\itshape}%
  {\medskip}
\newenvironment{definition}[1][]%
  {\refstepcounter{thm}\medskip\noindent\textbf{Definition~\thethm\if\relax\detokenize{#1}\relax\else\space(#1)\fi.}\enspace}%
  {\medskip}

\begin{document}

\preprint{APS/123-QED}
\title{Detecting entanglement of non-Gaussian continuous-variable states from single-copy homodyne measurements}

\author{Moritz Straeter}
\affiliation{%
 Centre for Quantum Technologies, National University of Singapore, 117543 Singapore, Singapore
}%
\author{Michael Tsesmelis}%
\affiliation{%
 Centre for Quantum Technologies, National University of Singapore, 117543 Singapore, Singapore
}%
\author{Leong-Chuan Kwek}%
\affiliation{%
 Centre for Quantum Technologies, National University of Singapore, 117543 Singapore, Singapore
}%
\affiliation{%
 MajuLab, CNRS-UNS-NUS-NTU International Joint Research Unit, UMI 3654, Singapore
}%
\affiliation{%
 National Institute of Education, Nanyang Technological University, 1 Nanyang Walk, Singapore 637616
}%

\date{\today}

\begin{abstract}
The entanglement of Gaussian continuous-variable (CV) states is fully determined by the state's second moments. In contrast, some entangled non-Gaussian states evade every second-moment criterion, and non-Gaussian entanglement detection remains an experimental challenge. The $p_3$-PPT criterion detects entanglement using moments of the partial transpose of the density matrix. This criterion was recently extended to CV systems using photon-number-resolving detectors and multi-copy interferometry; here we introduce a single-copy homodyne protocol that detects bipartite CV entanglement via the same criterion. Unbiased U-statistic estimators for the partial-transpose moments $p_2$ and $p_3$ are constructed directly from randomized homodyne data and used to evaluate the $p_3$-PPT entanglement witnesses: a linear one for detection, and a quadratic one whose violation yields a dimension-free lower bound on the entanglement negativity. The protocol estimates $p_2$ and $p_3$ up to additive error $\varepsilon$ at Fock cutoff $N$ from $O((N+1)^{14/3}/\varepsilon^2)$ measurements at fixed confidence. We demonstrate the protocol on six families of Gaussian and non-Gaussian states, reaching $95\%$ empirical one-sided detection probability from $\sim 10^3$ to $10^4$ homodyne measurements for states with $\bar{n} \approx 2$, placing non-Gaussian entanglement detection within reach of current homodyne experiments.
\end{abstract}

\maketitle

\section{Introduction}
\label{sec:introduction}

Non-Gaussian continuous-variable (CV) states are essential for universal quantum computation~\cite{lloydQuantumComputationContinuous1999, mariPositiveWignerFunctions2012, GPK} and fault-tolerant error correction~\cite{GPK, nisetNoGoTheoremGaussian2009}, yet certifying their entanglement remains an open experimental challenge \cite{Friis:2018nyl, walschaersNonGaussianQuantumStates2021}. For Gaussian states, second-moments suffice to determine entanglement \cite{boundentangledgaussian}. However, standard tools like the Duan~\cite{duanInseparabilityCriterionContinuous2000} and Simon~\cite{simonPeresHorodeckiSeparabilityCriterion2000} criteria that test second-moments can be blind to entanglement of non-Gaussian states, for example in photon-subtracted states \cite{walschaersEntanglementWignerFunction2017}, NOON states, and certain cat state regimes~\cite{NOONnotDetected}. The Shchukin--Vogel (SV) hierarchy~\cite{shchukinInseparabilityCriteriaContinuous2005} can in principle detect any state with negative partial transpose by testing positivity of moment matrices of increasing order, but the estimation cost grows rapidly with matrix dimension \cite{kanari-naishOptimizingConfidenceNegativepartialtransposebased2025}.

A more targeted approach uses low-order moments of the partially transposed state. The $p_3$-PPT criterion \cite{elbenMixedStateEntanglementLocal2020a, elbenRandomizedMeasurementToolbox2022, yuOptimalEntanglementCertification2021} certifies entanglement from just the second and third moment, and was recently extended to CV systems by Deside \emph{et al.} \cite{deside2025detectinggenuinenongaussianentanglement}, who access these moments via passive linear interferometers applied to two and three simultaneous copies of the state, followed by photon-number-resolving (PNR) detection. While conceptually clean, both ingredients are experimentally demanding. PNR detection capabilities remain less widespread in CV optics laboratories than homodyne setups. Furthermore, for the heralded non-Gaussian states, the multi-copy requirement is severe. The triple-coincidence rate is suppressed by orders of magnitude relative to the single-copy generation rate. In their discussion, Deside \emph{et al.} themselves identify randomized measurements \cite{elbenMixedStateEntanglementLocal2020a,elbenRandomizedMeasurementToolbox2022, cieslinskiAnalysingQuantumSystems2024a, Huang2020}, which were recently extended to CV systems via heterodyne and homodyne shadow tomography \cite{beckerClassicalShadowTomography2024, gandhariPrecisionBoundsContinuousVariable2024},  as a single-copy alternative worth investigating  in the CV regime. 

In this work, we develop such a single-copy protocol. We construct unbiased U-statistic estimators for the partial-transpose (PT) moments $p_2$ and $p_3$ directly from randomized homodyne data, exploiting the informational completeness of homodyne pattern functions on the truncated Fock space. These estimators access PT moments directly, bypassing full state tomography and avoiding the downward-biased eigenvalue estimates and confidence-set calibration that density-matrix reconstruction would require \cite{christandl,lvovskyContinuousvariableOpticalQuantumstate2009}. The moments are used to evaluate two $p_3$-PPT witnesses: a linear form \cite{deside2025detectinggenuinenongaussianentanglement} $W_\text{lin} = p_3 -(3p_2 - 1)/2$ for entanglement detection, and a quadratic form~\cite{elbenMixedStateEntanglementLocal2020a} $W_\text{quad} = p_3 - p_2^2$ whose violation additionally yields a quantitative negativity bound.

Applied to six families of entangled states \textemdash including photon-subtracted, photon-added, NOON, cat, and two-mode squeezed vacuum (TMSV) states, as well as a mixed non-Gaussian benchmark \textemdash the protocol achieves $95\%$ (one-sided) certification from less than $10^4$ homodyne measurements for states with $\bar{n} \sim 2$. Four structural properties underpin the protocol. First, it requires only single-copy homodyne detection with randomized phases and no multi-copy interferometry or PNR detectors. This makes it more easily accessible for many laboratories and eliminates the triple-coincidence bottleneck. Second, a self-certification truncation guarantee ensures that entanglement at any Fock cutoff $N$ is genuine: the truncated witness $W^{(N)}_\text{quad/lin} < 0$ implies $\rho^{T_B} \not\geq0$ for the full infinite-dimensional state, so no false positives arise from truncation. Third, local optical loss cannot create false-positive detection, so no detector-efficiency calibration is needed for detection validity. Fourth, the same estimated moments $(p_2, p_3)$ yield a dimension-free lower bound on the entanglement negativity, providing quantification beyond binary detection.

We focus on the $p_3$-PPT witness because it admits closed-form inference and a direct negativity bound (Sec.~\ref{sec:negativity_bound}). Stronger post-processing, e.g. the $p_3$-OPPT \cite{yuOptimalEntanglementCertification2021} and full truncated-state PPT testing \cite{peresSeparabilityCriterionDensity1996, shchukinInseparabilityCriteriaContinuous2005}, can be applied to the same homodyne data in principle.
\section{Protocol and Main Results}
\label{sec:protocol}

The protocol consists of three steps: (1)~collecting randomized homodyne data from single copies of the state, (2)~estimating the PT moments $p_2$ and $p_3$ via U-statistics constructed from homodyne shadow snapshots, and (3)~testing the linear witness $W_{\text{lin}} = p_3 - \frac{3p_2 -1}{2} \geq 0$. We now develop each component and establish the theoretical guarantees.

This section is organized as follows.
Section~\ref{sec:shadow_snapshots} introduces with the homodyne shadow framework of Gandhari \emph{et al.}\cite{gandhariPrecisionBoundsContinuousVariable2024} and the $p_3$-PPT entanglement witnesses the theoretical prerequisites for our method. We then show how to construct unbiased estimators in Sec.~\ref{sec:estimators_and_properties} and that the witness remains valid under Fock space truncation. We introduce the entanglement detection rule used throughout the manuscript (Sec.~\ref{sec:certification_rules}) and derive a dimension-free negativity lower bound (Sec.~\ref{sec:negativity_bound}). Sec.~\ref{sec:robustness} briefly comments on the robustness of the protocol and
Sec.~\ref{sec:recipe} summarizes the full protocol as an experimental
recipe.

\subsection{Theoretical foundations}
\label{sec:shadow_snapshots}

We consider a bipartite bosonic state~$\rho$ of finite mean photon
number, $\mathrm{tr}[\rho\hat{n}] = \bar{n} < \infty$, on which
only single-copy homodyne measurements are available.  In each
experimental run, we measure both modes by independent homodyne
detection at uniformly random phases
$\theta_A, \theta_B \sim \mathrm{Unif}[-\pi/2, \pi/2)$
(Fig.~\ref{fig:measurement_scheme}). A homodyne measurement measures a rotated quadrature operator $\hat{x}_{\theta} = (\hat{a} e^{-i\theta} + \hat{a}^\dagger e^{i\theta})/\sqrt{2}$, where $\hat{a},\hat{a}^\dagger$ are bosonic annihilation and creation operators ($[\hat{a}, \hat{a}^\dagger] = 1$). The rotation $\theta$ is realized through a local-oscillator phase \cite{Leonhardt1997-bs}. The measurement outcomes are $x \in \mathbb{R}$ real and sampled from the homodyne distribution $x\sim{p}_\rho(x|\theta) = \bra{x_\theta}\rho\ket{x_\theta}$, or in the bipartite case 
\begin{align} 
\label{eq:homod_distr}
x_A, x_B \sim {p}_\rho&(x_A, x_B | \theta_A, \theta_B) \notag \\ &= \bra{x_{\theta_A}, x_{\theta_B}}\rho\ket{x_{\theta_A}, x_{\theta_B}}. 
\end{align} Correlations, such as entanglement, will enter through correlations of the homodyne distribution (Eq.~\ref{eq:homod_distr}). In the infinite-sample limit,
the measurement record determines all Fock-basis matrix elements
of~$\rho$~\cite{lvovskyContinuousvariableOpticalQuantumstate2009, Gill2003InvitationQT}.

\begin{figure}[!htb]
\includegraphics[width=0.8\columnwidth]{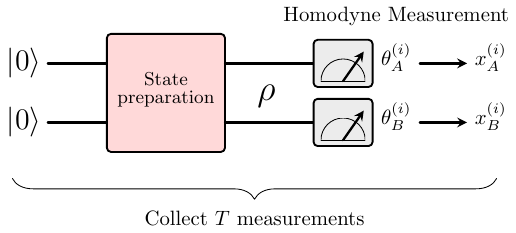}
\caption{\label{fig:measurement_scheme} Schematic of the measurement
  scheme.  After the state~$\rho$ is prepared, a homodyne measurement
  is performed on each mode with independent measurement phases
  $\theta^{(A)}_i, \theta^{(B)}_i$ uniformly sampled from
  $[-\pi/2,\pi/2)$.}
\end{figure}

Following the homodyne-shadow framework of Gandhari \emph{et
al.}~\cite{gandhariPrecisionBoundsContinuousVariable2024}, we map each
outcome $(x^{(i)}, \theta^{(i)})$ to a shadow snapshot $\rho_N^{(i)}$ in the truncated Fock space $\{\ket{0},\ldots \ket{N}\}$. Throughout, $N$ denotes the maximum Fock index per mode, so the truncated Hilbert space has dimension $N+1$ per mode. The homodyne shadow snapshot is given by
\begin{equation}
\label{eq:single_mode_shadow}
  \rho_N^{(i)}
  = \sum_{n,m=0}^{N}
    F_{n,m}\!\left(x^{(i)},\theta^{(i)}\right)
    \ket{n}\!\bra{m},
\end{equation}
where $F_{n,m}$ are the homodyne pattern
functions~\cite{Leonhardt1997-bs, richterRealisticPatternFunctions2000}.
An individual snapshot is not necessarily a valid physical density operator, but
its expectation satisfies $\mathbb{E}[\rho^{(i)}] = \rho_N := \Pi_N\rho\Pi_N$,
where $\Pi_N$ is the single-mode Fock projector. The estimators
are thus unbiased for the truncated state.  We show in
Sec.~\ref{sec:estimators_and_properties} that entanglement detected at any
truncation~$N$ is genuine.  Explicit expressions for the pattern functions
and their properties are given in the Supplemental Material Sec.~SI.

The positive-partial-transpose (PPT) criterion states that separable states necessarily have a partial-transpose with no negative eigenvalues ($\rho^{T_B} \geq 0$) \cite{peresSeparabilityCriterionDensity1996, simonPeresHorodeckiSeparabilityCriterion2000}. If we can show that the partial transpose of a quantum state has at least one negative eigenvalue, it must thus be entangled. However, PPT testing is challenging in practice because reconstructing $\rho^{T_B}$ introduces finite-sample bias in its eigenvalues. Although the homodyne shadows are informationally complete on the truncated Fock space, direct reconstruction and PPT testing would introduce downward-biased eigenvalue estimates and require confidence-set calibration~\cite{yuOptimalEntanglementCertification2021}; the scalar witnesses introduced below avoid these complications.

The partial-transpose (PT) moment of order-$n$ of a bipartite state~$\rho$ is defined as
\begin{equation}
\label{eq:pt_moments}
  p_n = \operatorname{Tr}\left[(\rho^{T_B})^n\right],
\end{equation}
where $(\,\cdot\,)^{T_B}$ denotes partial transposition on
subsystem~$B$.  If $\rho^{T_B} \geq 0$, the moments $\{p_n\}$
satisfy the same inequalities as the power sums of a probability
distribution; the leading non-trivial
condition, given $p_1 = \operatorname{Tr}\left[\rho\right]=1$, is the $p_3$-PPT
witness~\cite{elbenMixedStateEntanglementLocal2020a,
yuOptimalEntanglementCertification2021}:
\begin{equation}
\label{eq:p3PPT}
  W _\text{quad}= p_3 - p_2^2 \geq 0.
\end{equation}
If $W_\text{quad} < 0$, the state's partial-transpose necessarily has negative eigenvalues  (NPT) and the state is therefore entangled. The criterion $W_\text{quad} < 0$ is a Cauchy--Schwarz relaxation of the stronger optimal $p_3$-PPT criterion ($p_3$-OPPT) boundary~\cite{yuOptimalEntanglementCertification2021}. Ref.~\cite{deside2025detectinggenuinenongaussianentanglement} introduced the linear criterion 
\begin{equation}
\label{eq:Wlin}
    W_\text{lin} = p_3  - \frac{3p_2 -1}{2} \geq 0,
\end{equation}
which is stronger than $W_\text{quad}$ for $p_2 \in [\frac{1}{2}, 1)$ and uses the same partial transpose moments $p_2$ and $p_3$. We use the linear criterion as primary detection witness; if the quadratic witness is also negative, $(p_2,p_3)$ imply a lower bound on the state's negativity  (Sec.~\ref{sec:negativity_bound}).

For discrete-variable quantum systems $p_2$ and $p_3$ can be expressed as multi-copy expectation values of SWAP and cyclic-shift operators respectively \cite{elbenMixedStateEntanglementLocal2020a, elbenRandomizedMeasurementToolbox2022}.
The second moment $p_2$, which is equal to the states purity, is evaluated from two copies via
the SWAP trick~\cite{ekertDirectEstimationsLinear2002}:
\begin{equation}
\label{eq:SWAP_trick}
  p_2 = \operatorname{Tr}\left[(\rho \otimes \rho) S\right],
\end{equation}
where $S = S_{AA'}\!\otimes S_{BB'}$ is the bipartite SWAP
operator.  Here we use the identity $p_2 = \mathrm{Tr}[(\rho^{T_B})^2] = \mathrm{Tr}[\rho^2]$, which holds because the Hilbert--Schmidt norm is invariant under partial transposition.  The third PT-moment uses a cyclic-shift
decomposition~\cite{zhouSingleCopiesEstimationEntanglement2020}:
\begin{equation}
\label{eq:cyclic_shift}
  p_3 = \operatorname{Tr}\left[
    \rho^{\otimes 3}
    \left(
      \Pi_{(123)}^{(A)}
      \otimes \Pi_{(132)}^{(B)}
    \right)
  \right],
\end{equation}
where $\Pi_{\pi}^{(\cdot)}$ denotes the cyclic-shift operator acting on subsystem~$(\,\cdot\,)$, with $\pi \in S_3$ and $S_3$ being the permutation group.
\subsection{PT-Moment Estimators}
\label{sec:estimators_and_properties}

We now derive unbiased estimators for $p_2$ and $p_3$ from randomized homodyne measurements. Equations \eqref{eq:SWAP_trick} and \eqref{eq:cyclic_shift} express $p_2$ and $p_3$ as expectation values on two and three copies of $\rho$ respectively. Our goal is to estimate these moments without requiring multi-copy access from homodyne measurements. We estimate these multi-copy quantities from single-copy data using \emph{U-statistics}~\cite{Hoeffding1963, Serfling1980}: the average of a fixed symmetric kernel over all distinct $k$-tuples of the data is, for independent and i.i.d. distributed runs the minimum-variance unbiased estimator for the underlying functional. U-Statistics are the established route to PT-moments in discrete-variable randomized measurements \cite{elbenMixedStateEntanglementLocal2020a, elbenRandomizedMeasurementToolbox2022};l here we adapt them to homodyne shadows.

We write the tuple of measurement angles and corresponding outcomes as $X_i = (\theta^{(i)}_A,\theta^{(i)}_B,x^{(i)}_A,x^{(i)}_B)$ and let $\rho^{(i)}_M$ denote the single-mode shadow of mode $M\in\{A,B\}$ introduced in Sec. \ref{sec:shadow_snapshots}. Throughout, a hat marks an estimator and a superscript $(N)$ a quantity of the truncated Fock state. Because the pattern function act on each of the two modes separately, the bipartite snapshot is a product state $\rho^{(i)} = \rho^{(i)}_A \otimes \rho_B^{(i)}$, even when $\rho$ is entangled. Correlations enter through jointly sampling of $(x_A^{(i)}, x_B^{(i)})$ from the homodyne distribution of the state (Eq.~\ref{eq:homod_distr}), in expectation we recover correlations $\mathbb{E}[\rho_A^{(i)} \otimes \rho_B^{(i)}] = \rho_N$. To estimate the tensor product $\rho_N^{\otimes k}$, we note that the tensor product factorizes and may replace each $\rho_N$ by a homodyne snapshot from $k$ distinct and independent measurement,
\begin{align} 
\mathbb{E}\left[\rho^{(i_1)} \otimes \cdots \otimes \rho^{(i_k)}\right] \qquad i_1,\cdots, i_k \text{ distinct.}
\end{align}
In the multi-copy traces of Eq.~\eqref{eq:SWAP_trick} and Eq.~\eqref{eq:cyclic_shift} the same applies without changing the expectation. For the SWAP expectation value, since $S=S_{AA'} \otimes S_{BB'}$ swaps each mode separately, the trace factorizes 
\begin{align} 
\operatorname{Tr}\left[\left(\rho^{(i)} \otimes \rho^{(j)}\right) S\right] =& \, \operatorname{Tr}\left[\rho_A^{(i)}\rho_A^{(j)}\right]\operatorname{Tr}\left[\rho_B^{(i)}\rho_B^{(j)}\right] \\
&=: h_2(X_i, X_j),
\end{align}
with $\mathbb{E}\left[ h_2 \right] = p_2^{(N)}$ for $i\neq j$. In the Fock basis each factor is a sum over the pattern functions of the two runs:
\begin{align} \operatorname{Tr}\left[\rho_M^{(i)} \rho_M^{(j)}\right] = \sum_{n,m=0}^N F_{n,m}\left(x_M^{(i)}, \theta_M^{(i)}\right) F_{m,n}\left( x_M^{(j)}, \theta_M^{(j)} \right).\end{align}
This kernel is symmetric in its arguments, and its average over the $\binom{T}{2}$ distinct measurement outcome pairs is the unbiased U-statistic estimator 
\begin{align} \hat{p}_2^{(N)} = \binom{T}{2}^{-1} \sum_{i <j} h_2\left(X_i, X_j\right), \quad \mathbb{E} \left[\hat{p}_2^{(N)}\right] = p_2^{(N)}.\end{align}

The third moment (Eq~\eqref{eq:cyclic_shift}) uses three snapshots in the cyclic-shift expectation value. The forward shift ${\Pi}_{(123)}^{(A)}$ orders the $A$ factors as $i\rightarrow j\rightarrow k$, the reverse shift ${\Pi}_{(132)}^{(B)}$ orders the $B$ factors oppositely. so the kernel splits into a forward and reversed single-mode trace,
\begin{align} h_3\left( X_i, X_j, X_k \right) = \operatorname{Tr}\left[\rho_A^{(j)} \rho_A^{(j)} \rho_A^{(k)}\right]\operatorname{Tr}\left[ \rho_B^{(k)} \rho_B^{(j)} \rho_B^{(i)}\right],\end{align} with mean $p_3^{(N)}$ on distinct indices. The reversed ordering on $B$ implements the partial transpose on that subsystem; the explicit Fock-index contraction is given in Appendix~\ref{app:kernels}. The single-mode shadows are Hermitian, since pattern functions obey $F_{m,n} = F_{n,m}^{*}$~\cite{Leonhardt1997-bs}, so interchanging two arguments conjugate the kernel $h_3(X_i, X_k, X_j) = h_3(X_i, X_j, X_k)^{*}$. The three cyclic permutations $(i,j,k)$ therefore leave $h_3$ unchanged, while the three anti-cyclic ones conjugate it. This means the symmetrized kernel $\tilde{h}_3$ is exactly the real part $\tilde{h}_3 = \mathrm{Re}[h_3]$. Taking the real part preserves the mean as $p_3^{(N)}$ since the truncated third PT-moment is real as well. The average over the $\binom{T}{3}$ distinct triples is the unbiased U-statistic estimator \begin{align} \hat{p}_3^{(N)} = \binom{T}{3}^{-1}   \sum_{i<j<k} \tilde{h}_3 (X_i,X_j,X_k), \quad \mathbb{E}[\hat{p}_3^{(N)}] = p_3^{(N)}.\end{align}
Although expressed as sums over $\binom{T}{k}$ tuples, both estimators can be evaluated in closed form from accumulated single-run matrices, at a computational cost linear in the number of runs (Appendix \ref{app:kernels}).

Since the estimators target the truncated state's moments $p_k^{(N)} = \mathrm{tr}[(\rho_N^{T_B})^k]$, we evaluate truncated linear and quadratic witnesses for the subnormalized state $\operatorname{Tr}\left[\rho_N\right] = t_N$, which is given by \begin{align} 
&W_{\text{lin}}^{(N)} = p_3^{(N)} - \frac{3t_Np_2^{(N)} - t_N^3}{2} \\
&W_{\text{quad}}^{(N)} = t_N p_3^{(N)} - (p_2^{(N)})^2.
\end{align}

\emph{Proposition.---}%
The truncated witness remains valid for entanglement detection. For any bipartite state~$\rho$ and any Fock cutoff~$N \geq 1$,
\begin{equation}
\label{eq:trunc_chain}
  W^{(N)}_{\text{lin/quad}} < 0
  \Longrightarrow
  (\rho_N)^{T_B} \not\geq 0
  \Longrightarrow
  \rho^{T_B} \not\geq 0
\end{equation}
\emph{Proof.---} See Appendix \ref{app:truncation}.\\

Since $\rho^{T_B}<0$ implies entanglement of the non-truncated state $\rho$, the truncated witnesses are valid entanglement witnesses for the underlying state. Setting $t_N=1$ yields still a valid but weaker witness; we adopt $t_N=1$ throughout the manuscript, trading detection power for reduced variance and statistical convenience, and mitigate the loss by choosing $N$ such that $t_N = \operatorname{tr}[\rho_N] \approx 1$.

\subsection{Entanglement detection rules}
\label{sec:certification_rules}
Since $W^{(N)}_\text{lin} <0$ implies entanglement at any Fock cutoff $N$, entanglement detection reduces to deciding, from finitely many runs, whether the witness is negative. From the moment estimators of the previous section, we form 
\begin{align} 
\hat{W}_\text{lin}^{(N)} = \hat{p}_3^{(N)} - \frac{1}{2} (3 \hat{p}_2^{(N)} - 1),
\end{align}
which is an unbiased estimator for $W^{(N)}_\text{lin}$. We declare entanglement detection when its one-sided upper confidence bound is negative. The probability of this at $T$ runs,
\begin{equation}\label{eq:pi_alpha}
  \pi_\alpha(T) = \mathbf{Pr}\left[ \hat{W}_\text{lin}(T) + z_{1-\alpha}\hat{\sigma}_{W_\text{lin}}(T) < 0 \right],
\end{equation}
is the detection probability , with $z_{1-\alpha}$ the standard-normal quantile and $\hat{\sigma}_{W_\text{lin}}$ the standard error of $W_\text{lin}^{(N)}$ estimated of first-order Hoeffding projections \cite{Lee1990} of $\hat{p}_2^{(N)}$ and $\hat{p}_3^{(N)}$ (Supplemental Material Sec.~SIV). A separable state, for which $W^{(N)}_\text{lin} \geq 0$ can trigger the rule only through statistical fluctuations, with asymptotic probability of at most $\alpha$. We fix $\alpha = 0.05$ throughout. The experimental cost is the number of measurements at which detection becomes almost certain,
\begin{equation}\label{eq:Talpha95}
  T_\alpha^{95} = \min\bigl\{T : \pi_\alpha(T) \geq 0.95\bigr\}.
\end{equation}
The rule is studentized and asymptotically valid at level $\alpha$: its confidence coverage attains coverage $1-\alpha$ as $T\rightarrow \infty$ through the central-limit theorem of U-Statistics \cite{hoeffding48, Lee1990}. We calibrate its finite-$T$ behaviour via Monte Carlo in Sec.~\ref{sec:examples}. Across the six state families studied here, $T_{\alpha}^{95}$ lies between $10^3$ and $10^4$ homodyne runs, within reach of current homodyne experiments. 

Beyond the central-limit approximation, a distribution-free guarantee is available. Propagating per-moment guarantees (see Appendix \ref{app:finite-sample-guarantees}) through Eq.~\eqref{eq:Wlin}, yields the non-asymptotic certificate 
\begin{align}
\label{eq:rig_fallback}
\hat{W}_{\text{lin}}^{(N)} + \varepsilon < 0 \Rightarrow &{W}_{\text{lin}}^{(N)} < 0
\end{align}
with probability $1-\delta$, which holds for $T=O(N^{14/3}/\delta\varepsilon^2)$. This certifies entanglement without central-limit assumption (Appendix \ref{app:finite-sample-guarantees}). The Fock cutoff dependence disappears for states whose partial-transpose decay geometrically in the Fock basis. This includes for example two-mode squeezed vacuum (TMSV) states. For such states, we need $T=O(C_{\rho}/\delta\varepsilon^2)$ measurements for distribution free certification. These worst case bounds are conservative, requiring $T\gtrsim10^9$ runs against empirically $T_{\alpha}^{95} \sim 10^3-10^4$. We therefore use the studentized rule for the results of Sec.~\ref{sec:certification_rules} and use Eq.~\eqref{eq:rig_fallback} as a rigorous fallback.

\subsection{Entanglement Quantification}
\label{sec:negativity_bound}
The negativity~$\mathcal{N}(\rho) = (\|\rho^{T_B}\|_1 - 1)/2$ provides a way to quantify the amount of entanglement of the state \cite{VidalWerner2002, Plenio2005}. We know that $W_\text{quad} < 0$ and $W_\text{lin} < 0$ imply that the state is NPT and $\rho^{T_B}$ has at least one negative eigenvalue. Thus, the state has finite negativity ~$\mathcal{N}(\rho) >0$. PT-moments can be used to evaluate bounds on the states negativity \cite{carteret2017estimatingentanglementnegativityloworder,PhysRevA.91.022323}. Here, we present closed-form lower bounds on the states negativity from $(p_2,p_3)$. We derive a general bound~$\mathcal{N}(\rho) \geq \mathcal{N}_\text{lb}(p_2, p_3)$ and pure state bound $\mathcal{N}(\ket{\psi}\bra{\psi}) \geq \mathcal{N}_\text{pure}(p_2,p_3)$ from the root of a cubic polynomial. This provides entanglement quantification at no additional measurement cost.\\

\emph{Theorem.---}%
Let~$\rho$ be a bipartite state with PT-moments $p_1 = 1$, $p_2$, and~$p_3$ satisfying $p_3 < p_2^2$.  Then:

(i)~The negativity satisfies
$\mathcal{N}(\rho) \geq \beta_*$, where $\beta_*$ is the smallest
positive real root of
\begin{equation}
\label{eq:neg_cubic}
  g(u) = u^3 + 2p_2u^2 + p_3u + (p_3 - p_2^2).
\end{equation}

(ii)~The negativity also satisfies the closed-form rational bound
\begin{equation}
\label{eq:neg_rational}
  \mathcal{N}(\rho)
  \geq
  \frac{p_2^2 - p_3}{p_2 + p_3 + \tfrac{1}{4}} =\mathcal{N}_{\mathrm{lb}},
\end{equation}
and (iii) for pure states
\begin{align} \label{eq:neg_pure}
\mathcal{N}(\ket{\psi}\bra{\psi}) \geq \frac{-1 + \sqrt{5-4p_3}}{2} = \mathcal{N}_{\mathrm{pure}}.
\end{align}
Both bounds are dimension-free and therefore directly applicable to
CV systems. 

\emph{Proof sketch.---}%
Decompose the spectrum of $\rho^{T_B}$ into positive eigenvalues $\{\alpha_j\}$ and negative eigenvalues $\{-\beta_k\}$, with $\mathcal{N} = \sum_k \beta_k$. Applying Cauchy--Schwarz to $\sum_j \alpha_j^3$ from below, and using $\beta_k \leq \mathcal{N}$ to bound $\sum_k \beta_k^3$ from above, yields the polynomial inequality $g(\mathcal{N}) \geq 0$. Since $g(0) = p_3 - p_2^2 < 0$, the smallest positive root $\beta_{*}$ of ~(\ref{eq:neg_cubic}) gives $\mathcal{N} \geq \beta_{*}$. The rational bound ~(\ref{eq:neg_rational}) follows from the tighter spectral constraint $\operatorname{spec}(\rho^{T_B}) \subseteq [-\frac{1}{2}, 1]$~\cite{pt_neg_eigvals}. For pure states, $p_2=1$ factors the cubic. The full proof is in Appendix \ref{app:negativity}.\\

A similar bound is available from the truncated moments $p_2^{(N)}, p_3^{(N)}$:
\begin{align} 
&\mathcal{N}(\rho) \geq \frac{(p_2^{(N)})^2 - t_N p_3^{(N)}}{t_N p_2^{(N)} + p_3^{(N)} + t_N^3/4}=: \mathcal{N}_{\text{lb}}^{(N)}, \\
&\mathcal{N}(\ket{\psi}\bra{\psi})\geq \frac{-t_N^2 + \sqrt{5 t_N^4 - 4 t_N p_3^{(N)}}}{2t_N} =: \mathcal{N}_{\mathrm{pure}}^{(N)}
\end{align}
To estimate $(p_2^{(N)})^2$ one can build an order-4 U-statistic or square the estimator $\hat{p}_2^{(N)}$;  the latter is biased,  $\mathbb{E}[(\hat{p}_2^{(N)})^2]  = (p_2^{(N)})^2 + \operatorname{Var}[\hat{p}_2^{(N)}]$, and can be debiased by subtracting the estimated variance: $(\hat{p}_2^{(N)})^2 \rightarrow (\hat{p}_2^{(N)})^2 - \hat{\operatorname{Var}}[\hat{p}_2^{(N)}]$. We use the latter option, yielding the bias-corrected quadratic witness
\begin{align} 
\label{eq:quad_bias_corrected}
\hat{W}_\text{quad}' = p_3^{(N)} - (\hat{p}_2^{(N)})^2 + \hat{\operatorname{Var}}[\hat{p}_2^{(N)}].
\end{align}
\subsection{Robustness under Imperfections}
\label{sec:robustness}
We distinguish two aspects of robustness to loss.
\emph{Validity:} local loss channels preserve separability, so if
$W^{(N)}\!\bigl((\mathcal{L}_{\eta_A} \otimes
\mathcal{L}_{\eta_B})(\rho)\bigr) < 0$, then $\rho$ is
entangled---no calibration of~$\eta$ is required, and the certificate covers
detector inefficiency, fiber loss, and spatial mode mismatch
simultaneously (Appendix~\ref{app:robustness}).
\emph{Detection power:} loss does weaken the witness signal and raises the required sample count, since $|W_{\text{lin/quad}}^{(N)}|$ decreases with decreasing efficiency~$\eta$ (see  the sample-count scaling in Sec.~\ref{sec:robustness_demo}).

A constant phase offset~$\theta_0$ in the local-oscillator phase has
exactly zero effect on the estimators, because the U-statistic kernels
depend only on angle differences in which~$\theta_0$ cancels.
Stochastic jitter
$\sigma_\theta \sim 20\,\mathrm{mrad}$ contributes a bias
$< 2\%$ of~$p_2$ in the worst case (typically smaller for low-energy states), negligible at
$T \sim 10^3$ (Appendix~\ref{app:robustness}). The continuous phase distribution can be replaced by a discrete grid: $K_\theta \geq 2N + 1$ equally spaced phase settings give exactly unbiased
estimation.  For $N = 6$, $K_\theta = 13$ settings suffice.
Since each snapshot uses an independent phase, the
integration factorizes: $K_\theta \geq 2N + 1$ per mode per
snapshot suffices, regardless of the U-statistic order.

\subsection{Experimental recipe}
\label{sec:recipe}

We summarize the complete protocol as a six-step experimental recipe.
\begin{enumerate}
\setlength{\itemsep}{2pt}
\item Choose a Fock cutoff~$N$, ideally $t_N \approx 1$ (e.g., from a pilot estimate of the photon-number distribution).
\item Sample local-oscillator phases uniformly from $[-\pi/2, \pi/2))$.
\item Collect $T$ randomized homodyne snapshots $X_i = (\theta^{(i)}_A,\theta^{(i)}_B,x^{(i)}_A,x^{(i)}_B)$.
\item Compute $\hat{p}_2^{(N)}, \hat{p}_3^{(N)}$ via the U-statistic estimators.
\item Evaluate the linear witness for entanglement detection (Eq.~\eqref{eq:Wlin})
\item Optionally, evaluate the bias-corrected quadratic witness $\hat{W}_\text{quad}'$~\eqref{eq:quad_bias_corrected}, if negative, evaluate negativity lower bound
\end{enumerate}

\section{Demonstration}
\label{sec:examples}

We benchmark the protocol on six families of entangled states, both Gaussian and non-Gaussian: two-mode squeezed vacuum, its photon-subtracted and photo-added variants, NOON-states, entangled coherent states, and the photon-subtracted state of Walschaers \emph{et al.} \cite{walschaersEntanglementWignerFunction2017}. We treat the last in detail first, as its entanglement evades every second-moment criterion. It is hence the hard to detect by conventional means.  For each family, we report the sample cost of detection and the negativity the moments certify. Finally, we study the impact of optical loss and phase noise.

\paragraph*{Numerical protocol.}
We simulate homodyne measurements for each state (Code available at \footnote{Code will be made available in a public repository upon publication.}).  We run the estimators at an $N$ such that $1-\operatorname{Tr}\left[\rho_N\right] <0.05$, reported per state in Table \ref{tab:witness_detection}. By Proposition \eqref{eq:trunc_chain}, entanglement certified at any $N$ is valid, however, choosing a sutiable $N$ increases statistical power. Throughout this section we use the studentized detection rule (Eq.\eqref{eq:pi_alpha}), and report the empirical detection probability $\pi_{\alpha}(T)$ and the budget $T_{\alpha}^{95}$, calibrated from $K=150$ Monte Carlo repetitions per state and with $\alpha$ fixed at $\alpha =0.05$. In an experiment where the state and the photon-statistics are unknown, we recommend a
two-stage procedure: estimate the photon-number distribution
from a pilot dataset (${\sim}\,10\%$ of the total budget) and
fix~$N$ before using the remaining data for entanglement detection.

\subsection{{Photon-subtracted twin vacuum state}}
The state of Walschaers \emph{et al.}~\cite{walschaersNonGaussianQuantumStates2021} is prepared by mixing two squeezed vacuum modes on a beamsplitter, subtracting a photon from one arm, and recombining with an identical beamsplitter (Fig.~\ref{fig:phot_subtr_prep}). We fix the squeezing strength at $r=0.5$ ($\bar{n} \approx 2$) throughout.  The photon-subtraction is engineered in a way, that the covariance matrix of the final state satisfies the Simon criterion \cite{walschaersEntanglementWignerFunction2017, simonPeresHorodeckiSeparabilityCriterion2000}. Thus, it looks separable to every Gaussian entanglement witness.
\begin{figure}[b]
\includegraphics[width=\columnwidth]{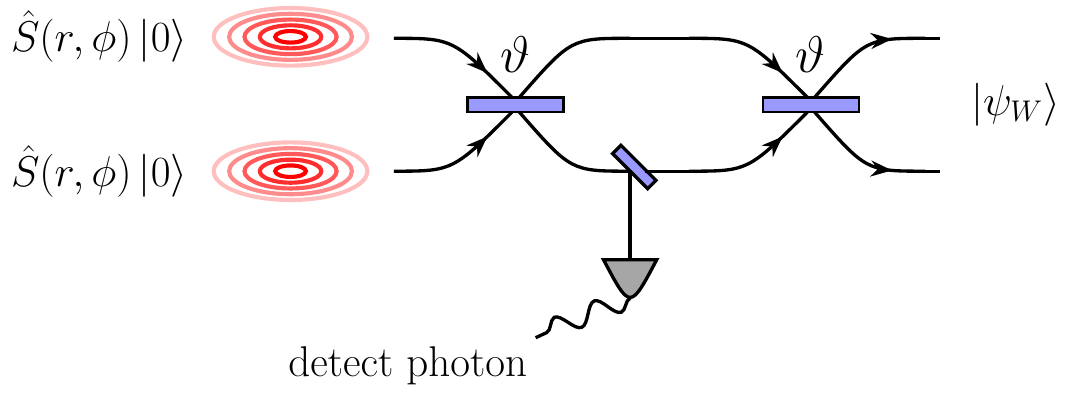}
\caption{\label{fig:phot_subtr_prep}
Schematic preparation of the photon-subtracted
state~\cite{walschaersNonGaussianQuantumStates2021}.
Two squeezed vacuum modes are mixed at beamsplitter angle~$\vartheta$;
a photon is subtracted via a high-transmissivity beamsplitter and single-photon
detector; a second beamsplitter yields the output state.}
\end{figure} The photon subtraction is designed so that the resulting covariance matrix satisfies the Simon criterion~\cite{simonPeresHorodeckiSeparabilityCriterion2000}; consequently, no Gaussian witness can detect the state's entanglement.
The $p_3$-PPT witness detects the states entanglement. Analytically, $W_\text{lin} = -\frac{3}{4}\operatorname{sin}^2\vartheta < 0$ for every $\vartheta \neq k\pi$. It remains to show, how many homodyne snapshots are necessary to resolve $W_\text{lin} < 0$  from single-copy homodyne data. Figure~\ref{fig:walschaer_convergence} shows the estimated moments $\hat{p}_{2}^{(N)}$ and $\hat{p}_{3}^{(N)}$, the witness $\hat{W}_{\text{lin}}^{(N)}$ and the negativity bound estimates $\hat{\mathcal{N}}^{(N)}_{\text{lb/pure}}$ as the number of homodyne snapshots $T$ increases, at $\vartheta=\pi/4$ and $r=0.5$. The U-statistics converge smoothly to their exact values , with truncation bias negligible against the statistical uncertainty up to $T \lesssim 10^4$.
\begin{figure}[htb!]
\includegraphics[width=\columnwidth]{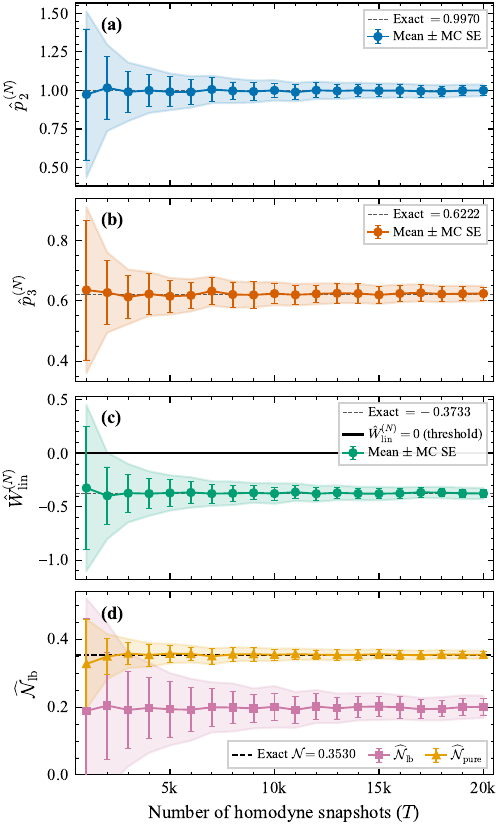}
\caption{\label{fig:walschaer_convergence}
Convergence of the estimated moments $\hat{p}_2^{(N)}$, $\hat{p}_3^{(N)}$,
and witness $\hat{W}_\text{lin}^{(N)}$ with the number of homodyne snapshots~$T$
for the photon-subtracted state at $\vartheta = \pi/4$, $r = 0.5$.
Shaded bands indicate mean Jackknife standard errors \cite{Arvesen1969} (see Supplemental for details), while error bars show Monte-Carlo standard errors. The estimates converge to the true values, the estimated and observed standard errors decrease with increasing $T$. (d) The pure state negativity bound is a tight estimate of the negativity; the general lower bound computed from the same moment estimates, is looser but a valid lower bound.}
\end{figure} 
Figure~\ref{fig:wal_pvalue} shows the detection probability $\pi_{\alpha}(T)$ for several squeezing strengths. The required sample count grows with squeezing, since a higher mean photon number needs a larger cutoff $N$ and the estimator variance grows with $N$. At $r=0.5$ ($\bar{n}\approx 2$) the protocol reaches $95\%$ one-sided detection at $T\approx 7{,}500$ homodyne snapshots.

\begin{figure}[!htb]
\includegraphics[width=\columnwidth]{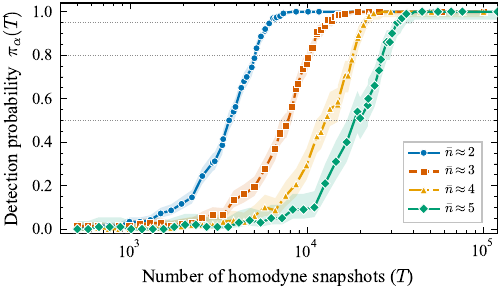}
\caption{\label{fig:wal_pvalue}
Empirical detection probability $\pi_\alpha(T)$ (Eq. ~(\ref{eq:pi_alpha})) ($\alpha =0.05$) as a function of the number of homodyne snapshots~$T$ for the
photon-subtracted state at $\vartheta = \pi/4$
and several squeezing strengths~$r$.
Higher squeezing requires more measurements due to increased
estimator variance at larger Fock truncation.}
\end{figure}

\paragraph*{Mixed non-Gaussian benchmark: noisy Walschaers state.}
To benchmark the protocol for mixed states,  we adopt a phenomenological vacuum-admixture benchmark for imperfect preparation of the Walschaers state:
\begin{equation}\label{eq:noisy_walschaers}
  \rho(\lambda)
  = (1 - \lambda)\,|\psi_W\rangle\langle\psi_W|
  + \lambda\,|0{,}0\rangle\langle 0{,}0|,
\end{equation}
where $|\psi_W\rangle$ is the Walschaers state ($r = 0.5$,
$\vartheta = \pi/4$) and $\lambda$ parametrizes the vacuum admixture.  At every noise level $\lambda$, the smallest symplectic eigenvalue is greater than $0.5$, so the Simon criterion cannot detect the entanglement. The $p_3$-PPT by contrast detects entanglement up to $\lambda \approx 0.53$ (purity $\sim 0.50$, negativity $\sim0.10$).  Figure~\ref{fig:mixed_state_walschaers} shows clear detection for $\lambda\lesssim0.25$.
\begin{figure}[t]
\includegraphics[width=\columnwidth]{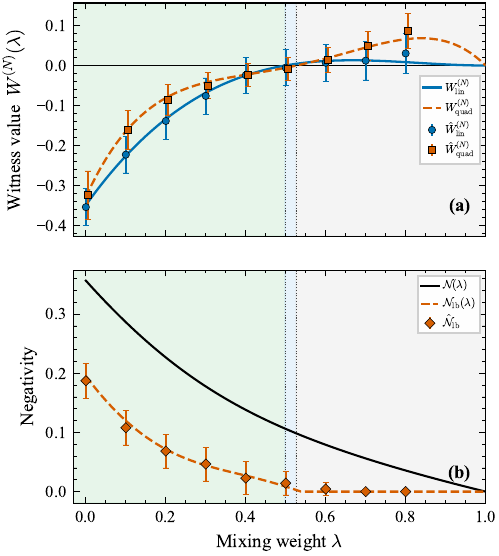}
\caption{\label{fig:mixed_state_walschaers}
(a) Witness value for the linear and quadratic witness of the mixed state (Eq.~\ref{eq:noisy_walschaers}) for different mixing weights $\lambda$. Lines indicate exact values, markers show data points, error bars are calculated from $K=150$ repetitions. In the green shaded region both witnesses detect entanglement in the light blue region only the quadratic witness does; in the grey region both witnesses are positive and detect nothing. (b) Negativity lower bound $\mathcal{N}_{\mathrm{lb}}$ for the same state. The black line shows the exact negativity, which remains $>0$ for $\lambda\in[0,1)$. The dashed red line shows the negativity lower bound, markers the estimated lower bound from estimated moments (error bars from $K=150$ repetitions).}
\end{figure}

\subsection{Multi-state survey}
\label{sec:multi_state}

Having investigated one demanding non-Gaussian state in detail, we now survey the protocol across a broader class of entangled states. We apply the protocol to five additional state families:
(i)~two-mode squeezed vacuum (TMSV), the canonical Gaussian
entangled state;
(ii)~photon-subtracted TMSV, obtained by subtracting $k$ photons
from one arm of a TMSV state \cite{walschaersNonGaussianQuantumStates2021};
(iii)~photon-added TMSV, with $k$ photons added instead \cite{Opatrny2000, NavarreteBenlloch2012};
(iv)~NOON states \cite{botonoon} $\propto \ket{n,0} + \ket{0,n}$,
maximally path-entangled Fock states;
(v)~entangled coherent (cat) states
$\propto \ket{\alpha,\alpha} + \ket{-\alpha,-\alpha}$ \cite{Sanders2012};
and (vi)~the Walschaers photon-subtracted state \cite{walschaersEntanglementWignerFunction2017} discussed above.

Figure~\ref{fig:gen_states_sign} shows the detection probability versus $T$ for representative parameters of each of the five added families. Table~\ref{tab:witness_detection} collects the key quantities for all six. Strongly entangled states are the cheapest to certify.  The $n=2$ NOON state reaches $95\%$ one-sided detection in only $T \approx 500$, thanks to its small Fock support and large witness magnitude. The remaining families require $T \approx 3{,}000$--$7{,}500$, the cost growing with mean photon number and Fock cutoff.

The negativity bound is tight for NOON
states (100\% of the exact negativity~$\mathcal{N}$) and cat
states ($\gtrsim 98\%$), but captures only 18--50\% for TMSV-type
states. The gap reflects the spectral structure of the partial transpose: NOON and cat states have a single dominant negative eigenvalue, making the two-moment bound nearly tight, whereas TMSV-family states spread their negativity across multiple eigenvalues. Tighter bounds from higher-order PT moments ($p_4, p_5$) would narrow this gap at the cost of increased estimator variance.

\begin{figure}[h]
\includegraphics[width=\columnwidth]{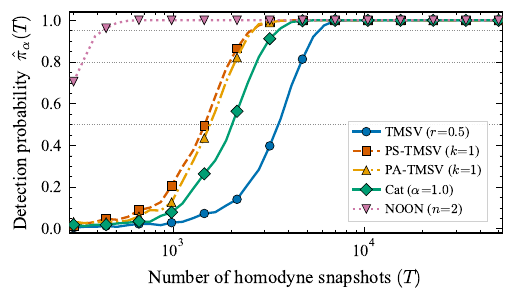}
\caption{\label{fig:gen_states_sign}
Empirical detection probability $\pi_\alpha(T)$ as a function of
homodyne snapshots~$T$ for five families of entangled states.
Each curve is averaged over $K = 150$ independent Monte-Carlo repetitions.}
\end{figure}

\begin{table*}[t]
\centering
\caption{Entanglement witness and negativity bounds across state families.
$\bar{n}$ is the mean total photon number, $N_{\mathrm{est}}$ the estimator Fock cutoff,
$W_{\text{lin}} = p_3 - (3p_2-1)/2$ the exact linear witness, $T_{\alpha}^{95}$ the homodyne snapshots required for
empirical (one-sided) detection ($K = 150$ repetitions), rounded up to the nearest $500$, $\mathcal{N}$ the exact negativity,
$\mathcal{N}_{\mathrm{lb}}$ the negativity lower bound of Sec.~\ref{sec:negativity_bound} (for pure states, the larger of the pure-state and rational bounds [Eqs.~\eqref{eq:neg_pure} and~\eqref{eq:neg_rational}]; cubic for mixed states), and the final column the tightness ratio.}
\label{tab:witness_detection}
\begin{ruledtabular}
\begin{tabular}{ld{1.2}cd{2.2}rccl}
\textrm{State} & \multicolumn{1}{c}{$\bar{n}$} & $N_{\mathrm{est}}$ & \multicolumn{1}{c}{$W_{\text{lin}}$} & \multicolumn{1}{c}{$T_{\alpha}^{95}$} & $\mathcal{N}$ & $\mathcal{N}_{\mathrm{lb}}$ & $\mathcal{N}_{\mathrm{lb}}/\mathcal{N}$ \\
\hline
\multicolumn{8}{l}{\textit{NOON states}} \\[2pt]
\text{$\ket{2,0}+\ket{0,2}$} & 2.00 & 3 & -0.75 & 500 & 0.50 & 0.50 & 100\% \\
\text{$\ket{3,0}+\ket{0,3}$} & 3.00 & 4 & -0.75 & 1,000 & 0.50 & 0.50 & 100\% \\
\text{$\ket{4,0}+\ket{0,4}$} & 4.00 & 5 & -0.75 & 2,000 & 0.50 & 0.50 & 100\% \\[4pt]
\multicolumn{8}{l}{\textit{Photon-subtracted TMSV} ($r=0.5$)} \\[2pt]
\text{$k=1$} & 1.98 & 6 & -0.88 & 3,000 & 1.70 & 0.65 & 38\% \\
\text{$k=2$} & 3.71 & 8 & -0.94 & 4,500 & 2.56 & 0.72 & 28\% \\
\text{$k=3$} & 5.46 & 10 & -0.96 & 6,000 & 3.28 & 0.75 & 23\% \\[4pt]
\multicolumn{8}{l}{\textit{Photon-added TMSV} ($r=0.3$)} \\[2pt]
\text{$k=1$} & 2.71 & 6 & -0.63 & 3,000 & 0.87 & 0.44 & 50\% \\
\text{$k=2$} & 5.45 & 8 & -0.84 & 3,000 & 1.35 & 0.60 & 44\% \\
\text{$k=3$} & 8.28 & 10 & -0.90 & 4,500 & 1.77 & 0.67 & 38\% \\[4pt]
\multicolumn{8}{l}{\textit{Two-mode squeezed vacuum}} \\[2pt]
\text{$r=0.3$} & 0.19 & 4 & -0.23 & 5,000 & 0.40 & 0.20 & 49\% \\
\text{$r=0.5$} & 0.54 & 6 & -0.51 & 6,500 & 0.83 & 0.37 & 45\% \\
\text{$r=0.7$} & 1.15 & 8 & -0.73 & 7,500 & 1.46 & 0.49 & 34\% \\[4pt]
\multicolumn{8}{l}{\textit{Entangled cat states}} \\[2pt]
\text{$\alpha=1.0$} & 1.93 & 6 & -0.69 & 4,000 & 0.48 & 0.47 & 98\% \\
\text{$\alpha=1.5$} & 4.50 & 8 & -0.74 & 6,000 & 0.50 & 0.50 & 100\% \\
\text{$\alpha=2.0$} & 8.00 & 10 & -0.71 & 9,500 & 0.49 & 0.49 & 100\% \\[4pt]
\multicolumn{8}{l}{\textit{Walschaers state} ($r=0.5$)} \\[2pt]
\text{$\vartheta=\pi/4$} & 2.09 & 6 & -0.35 & 7,000 & 0.35 & 0.28 & 80\% \\
\text{$\vartheta=3\pi/8$} & 2.09 & 8 & -0.63 & 7,500 & 0.46 & 0.44 & 96\% \\
\end{tabular}
\end{ruledtabular}
\end{table*}

\subsection{Robustness to loss and phase noise}
\label{sec:robustness_demo}
We finish the demonstration of the protocol by assessing the two dominant imperfections of homodyne experiments: optical loss and local-oscillator phase noise \cite{Leonhardt1997-bs}. Local detector inefficiencies cannot create false positives  (Sec.~\ref{sec:estimators_and_properties}). Any witness violation at any detector efficiency $\eta>0$ certifies entanglement of the source state without efficiency calibration. Here we characterize how loss and phase noise affect \emph{detection power}.

Figure~\ref{fig:robustness_heatmap} maps the witness value $W_\text{lin}(\eta, \sigma_\theta)$ over detector efficiency $\eta$ and phase-jitter standard deviation $\sigma_\theta$ for the photon-subtracted state Walschaers state. Also shown is the detection boundary $W_\text{lin} =0$ (and $W_{\text{quad}}$). Detector efficiency loss dominates, the $W_{\text{lin}} = 0 $ contours run nearly vertical, the boundary is mostly set by $\eta$, with phase jitter up to $\sigma_\theta \sim 300\,\text{mrad}$ playing only a secondary role. Other states are treated in Appendix~\ref{sec:app_loss}.
\begin{figure}[t]
\includegraphics[width=\columnwidth]{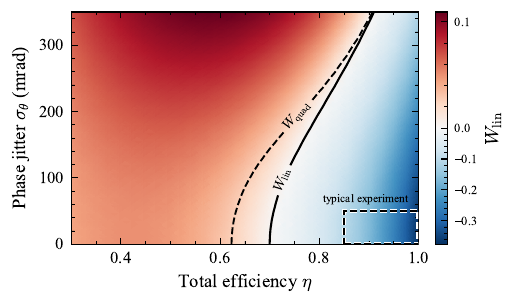}
\caption{\label{fig:robustness_heatmap}
Witness $W_{\text{lin}}(\eta,\sigma_\theta)$ and $W_{\text{quad}}(\eta,\sigma_\theta)$ for the photon-subtracted Walschaers state
(Fig.\ref{fig:phot_subtr_prep}) under imperfect detectors (efficiency $\eta$) and under homodyne measurement phase jitter $\sigma_\theta$. The colormap shows the witness value for the linear witness. The boundary is shown as a black (dashed) line for both the linear and quadratic witness.}
\end{figure}
Figure~\ref{fig:mixed_state} plots the sample cost $T_\alpha^{95}$ against detector efficiency for six state families. The count rises smoothly as $\eta$ falls, since loss reduces both purity and witness magnitude. For the photon-subtracted TMSV, $T_\alpha^{95}$ grows from $\sim 3{,}000$ at unit efficiency to $\sim 5{,}000$ at $\eta=0.9$ ($10\%$ loss). For TMSV, the sample cost grows from $\sim 7{,}500$ to $\sim12{,}000$ over the same range. NOON and entangled coherent states degrade more steeply. However, at realistic detector efficiencies of $\eta \approx 0.9$, the budgets remain of order $10^4$ runs, which is well within reach of current homodyne setups \cite{KUMAR20125259}.

shows the finite-sample cost $T_\alpha^{95}$ as a function of detector efficiency for six state families. The required sample count increases smoothly with decreasing $\eta$, since loss reduces both purity and witness magnitude. For the photon -subtracted TMSV, $T_\alpha^{95}$ grows from $\sim 3{,}000$ at ideal detection to $\sim 5{,}000$ at $\eta=0.9$ ($10\%$ loss); for the TMSV state, the cost increases from $\sim 7{,}500$ to $\sim 12{,}000$ over the same range. NOON and cat states degrade more steeply with loss, consistent with their higher critical-efficiency threshold. All the same counts remain within reach of current homodyne experiments operation at repetition rates of $10^5-10^6 \text{ Hz}$, confirming that the protocol is practical under realistic loss.
\begin{figure}[t]
\includegraphics[width=\columnwidth]{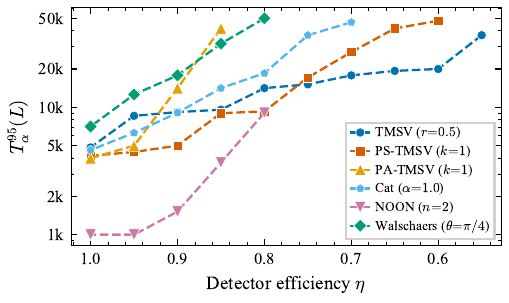}
\caption{\label{fig:mixed_state}
Sample count $T_{\alpha}^{95}$ required for $95\%$ empirical (one-sided) detection as a function of detector efficiency~$\eta$ for
TMSV ($r = 0.5$) and photon-subtracted TMSV ($r = 0.5$, $k = 1$), photon-added TMSV ($r =0.5, k=1$), NOON ($n=2$) and cat ($\alpha=1.0$) states.
Loss produces mixed states with $p_2 < 1$; the measurement cost
increases smoothly but remains experimentally feasible.}
\end{figure}

\section{Discussion \& Conclusion}
\label{sec:discussion}

The central message of this work is that non-Gaussian entanglement certification does not require non-Gaussian measurements. Standard homodyne detection, supplemented only by classical post-processing of randomized phase data suffices to estimate the partial-transpose moments $p_2$ and $p_3$ and to evaluate the linear and quadratic $p_3$-PPT witness. The sample counts reported are  between $10^3-10^4$ homodyne snapshots for the states with $\bar{n} \sim 2$ studied here. This places the protocol within reach of existing homodyne setups operating at kilohertz-to-megahertz repetition rates~\cite{KUMAR20125259}.

It is worth comparing our single-copy protocol with the multi-copy alternative of Ref.~\cite{deside2025detectinggenuinenongaussianentanglement}. The CV $p_3$-PPT extension of Deside \emph{et al.}~\cite{deside2025detectinggenuinenongaussianentanglement} accesses the same moments by interfering three simultaneous copies of $\rho$ through a Fourier-transform interferometer read out by photon-number-resolving detectors, attaining a lower variance per detection event than ours. That advantage is bought at a steep operational price. Three copies of the state must arrive in coincidence, suppressing the usable rate of heralded sources by order of magnitude. These copies must remain mutually coherent across a phase-stable interferometer and reproduce the target state with comparable fidelity. Single-copy operation removes all three requirements at once, each measurement acts on an  independently prepared copy. Furthermore, our protocol only relies on homodyne detection which is far more widely available than photon-number resolution. Our protocol complements the protocol of Deside \emph{et al.}, especially for laboratories who do not possess the hardware necessary for the multi-copy approach.

The price of single-copy operation is a polynomial variance scaling, $O\left((N+1)^{14/3}/\varepsilon^2\right)$ with the Fock cutoff~$N$, which is the protocol's principle limitation. In practice, it confines the method to roughly $\bar{n} \lesssim 3-4$ photons per mode ($N \lesssim 12$). This covers comfortably many experimentally relevant non-Gaussian states, but excludes strongly squeezed states ( $r>0.7$, above $6\,\mathrm{dB}$). For states whose Fock amplitudes decay geometrically, e.g. TMSV and photon-subtracted variants, this limitation is not as severe and the sample complexity becomes independent of $N$. We also introduced closed-form expressions for negativity lower-bounds from the second and third PT-moment. These bounds are dimension-independent and applicable to discrete-variable and continuous-variable quantum states.

Several extensions follow naturally. The most immediate is an experimental demonstration. Single-photon subtracted states from parametric down-conversion require only $\sim 3000$ heralded events, within reach of current experiments. Because the entanglement certificate is computed entirely in post-processing, recorded data can also support stronger criteria at no additional measurement cost. The optimal $p_3$-PPT criterion ($p_3$-OPPT) of Yu \emph{et al.} \cite{yuOptimalEntanglementCertification2021}, which is stronger than both the quadratic and the linear witness, can be evaluated from the same estimated moments. The framework also extends to multipartite settings, where a single dataset tests the PPT condition across very bipartition of an $M$-mode state and a union bound controls the family-wise error rate. Efficiency-dependent \emph{realistic} pattern functions~\cite{richterRealisticPatternFunctions2000} are available if we calibrate for the detector efficiency $\eta$. These reconstruct the pre-loss state and increase detection power. From the standpoint of quantum learning theory \cite{Huang2020}, a natural open question is the sample-optimal strategy for homodyne-based non-Gaussian entanglement detection, for which recent results on learning CV states  \cite{Wu2024,Mele2025, mele2026advancesquantumlearningtheory} can offer a starting point. 
\section*{Data and software availability}

No experimental data were generated or analyzed in this work. The numerical data supporting the figures and tables, together with the code used to generate the homodyne samples, evaluate the U-statistic estimators, perform the Monte Carlo analysis, and reproduce the plots, will be made available in a public repository upon publication.
\begin{acknowledgments}
Sone simulations of random homodyne measurements have been performed with the \textit{Python-}library \textit{Strawberryfields} \cite{Killoran2019, Bromley_2020}. M.S. and M.T. acknowledge the support of the Singapore National Quantum Scholarship Scheme (NQSS).  L.C.K acknowledges support from the National Research Foundation, Singapore, and the Ministry of Education, Singapore. We thank Celine My Thi Trieu for helpful input during the writing of this manuscript.
\end{acknowledgments}

\appendix

\section{Explicit kernel formulas and computational implementation}
\label{app:kernels}

We give the explicit Fock-basis expressions for the U-statistic
kernels used to estimate the PT-moments~$p_2$ and~$p_3$, including
the efficient accumulator implementation that avoids
the combinatorial sum over all tuples.

\subsection{Order-2 kernel: \texorpdfstring{$\hat{p}_2$}{p2-hat}}
\label{app:kernel_p2}

The kernel for the second PT-moment is
\begin{equation}\label{eq:h2_explicit}
  h_2(X_i, X_j)
  = g_A(X_i, X_j)\,g_B(X_i, X_j),
\end{equation}
where the single-mode factor is
\begin{align}\label{eq:gM_formula}
  g_M(X_i, X_j)
  &= \sum_{n,m=0}^{N}
    F_{n,m}(x_M^{(i)}, \theta_M^{(i)})\,
    F_{m,n}(x_M^{(j)}, \theta_M^{(j)})
\end{align}
Since $h_2$ is already symmetric under $i \leftrightarrow j$
(by cyclicity of the trace), the symmetrized kernel
$\tilde{h}_2 = h_2$, and the U-statistic is
\begin{equation}
  \hat{p}_2
  = \binom{T}{2}^{-1}\!\sum_{i < j} h_2(X_i, X_j).
\end{equation}

\paragraph{Efficient computation.}
Define the single-mode shadow matrices
$A_i := \rho_A^{(i)}$ and $B_i := \rho_B^{(i)}$
(each $(N\!+\!1) \times (N\!+\!1)$).  The bipartite shadow product in the
partial-transpose basis is the $(N\!+\!1)^2 \times (N\!+\!1)^2$ outer product
$P_i = A_i \otimes B_i^T$, reshaped as an $(N\!+\!1)^4$
tensor with indices
$[P_i]_{a,d,b,c} = [A_i]_{a,b}\,[B_i]_{c,d}$.
Accumulate the sum
\begin{equation}\label{eq:S_accumulate}
  S = \sum_{i=1}^T P_i
  = \sum_i A_i \otimes B_i^T,
\end{equation}
reshaped as an $(N{+}1)^2 \times (N{+}1)^2$ matrix.  Then
\begin{equation}\label{eq:p2_accumulate}
  \hat{p}_2
  = \frac{\mathrm{Tr}[S^2] - \sum_i \mathrm{Tr}[P_i^2]}
         {T(T - 1)},
\end{equation}
where $\mathrm{Tr}[P_i^2]
= \mathrm{Tr}[A_i^2]\,\mathrm{Tr}[B_i^2]$.
This is the standard diagonal-subtracted form: the coincident $i=j$
terms are removed so that only distinct snapshot pairs contribute,
keeping the estimator unbiased.
The cost is $O(T\,N^4)$ for the accumulation
plus $O(N^4)$ for $\mathrm{Tr}[S^2]$.

\paragraph{Hoeffding first projection.}
For the variance estimate, the first Hoeffding projection
is
\begin{equation}\label{eq:h1_p2}
  G_2(X_i) = \hat{p}_{2,1}^{(i)} - \hat{p}_2,
  \qquad
  \hat{p}_{2,1}^{(i)}
  = \frac{\mathrm{Tr}[P_i\,S] - \mathrm{Tr}[P_i^2]}{T - 1},
\end{equation}
and the standard error is
$\mathrm{SE}(\hat{p}_2)
= 2\,\mathrm{std}(G_2) / \sqrt{T}$.

\subsection{Order-3 kernel: \texorpdfstring{$\hat{p}_3$}{p3-hat}}
\label{app:kernel_p3}

The kernel for the third PT-moment uses the opposite-chirality
cyclic
shift [Eq.~\eqref{eq:cyclic_shift}]:
\begin{equation}\label{eq:h3_explicit}
  h_3(X_i, X_j, X_k)
  = \mathrm{Tr}\!\Bigl[
    \bigl(\rho_{AB}^{(i)} \otimes \rho_{AB}^{(j)}
    \otimes \rho_{AB}^{(k)}\bigr)
    \bigl(\hat\Pi_{(123)}^{(A)}
    \otimes \hat\Pi_{(132)}^{(B)}\bigr)
  \Bigr].
\end{equation}
Writing out the index contraction:
\begin{equation}\label{eq:h3_fock}
\begin{aligned}
  h_3(X_i,X_j,X_k)
  &= \sum_{\substack{a_1,a_2,a_3 \\ b_1,b_2,b_3}}
    [A_i]_{a_1,a_2}\,[B_i]_{b_1,b_3}\,
    [A_j]_{a_2,a_3}\,[B_j]_{b_2,b_1}\\
  &\qquad\qquad\qquad\times
    [A_k]_{a_3,a_1}\,[B_k]_{b_3,b_2}\\[4pt]
  &= \mathrm{Tr}[A_i\,A_j\,A_k]\;\mathrm{Tr}[B_i^T B_j^T B_k^T].
\end{aligned}
\end{equation}
The last factorization follows because the forward cyclic
shift on~$A$ and the reverse cyclic shift on~$B$ decouple:
the $A$-indices contract as a cyclic trace
$\sum_{a_1,a_2,a_3}[A_i]_{a_1,a_2}[A_j]_{a_2,a_3}
[A_k]_{a_3,a_1} = \mathrm{Tr}[A_i A_j A_k]$,
and the $B$-indices contract with the opposite cycle
$\sum_{b_1,b_2,b_3}[B_i]_{b_1,b_3}[B_j]_{b_2,b_1}
[B_k]_{b_3,b_2} = \mathrm{Tr}[B_k B_j B_i]
= \mathrm{Tr}[B_i^T B_j^T B_k^T]$.

The symmetrized kernel averages over the $3! = 6$
permutations of $(i,j,k)$.  By cyclic invariance of
the trace ($\mathrm{Tr}[ABC] = \mathrm{Tr}[BCA]$),
the three cyclic permutations give the same value, so
only two distinct terms appear:
\begin{equation}\label{eq:h3_sym}
  \tilde{h}_3(X_i,X_j,X_k)
  = \tfrac{1}{2}\bigl[
    h_3(X_i,X_j,X_k)
    + h_3(X_i,X_k,X_j)
  \bigr].
\end{equation}
Since $A_i$ and $B_i^T$ are Hermitian, swapping the last two operands in each trace acts as complex conjugation: $\operatorname{Tr}[A_i A_k A_j] = \operatorname{Tr}[A_i A_j A_k]^*$ and likewise for $B_i^T$. Hence, $h_3(X_i, X_k, X_j) = h_3(X_i, X_j, X_k)^*$ and the symmetrized kernel is real-valued by construction, $\tilde{h}_3 = \operatorname{Re}[h_3]$.
\paragraph{Efficient computation.}
Using the accumulated matrix~$S$~\eqref{eq:S_accumulate}
(reshaped to $(N{+}1)^2 \times (N{+}1)^2$), define
\begin{equation}\label{eq:Q_def}
  Q = \sum_{i=1}^T P_i^2
  = \sum_i (A_i^2) \otimes (B_i^2)^T,
\end{equation}
and the diagonal cubic sum
$D = \sum_i \mathrm{Tr}[A_i^3]\,\mathrm{Tr}[(B_i^T)^3]$.
Then the inclusion-exclusion formula gives
\begin{equation}\label{eq:p3_accumulate}
  \hat{p}_3
  = \frac{\mathrm{Tr}[S^3]
    - 3\,\mathrm{Tr}[Q\,S]
    + 2\,D}
    {T(T-1)(T-2)},
\end{equation}
where $\mathrm{Tr}[S^3] = \mathrm{Tr}[(S^2)\,S]$ involves
only one additional matrix multiplication.
The total cost is $O(T\,N^4 + N^6)$ time and
$O(N^4)$ memory.

\paragraph{Hoeffding first projection.}
\begin{equation}\label{eq:h1_p3}
  G_3(X_i)
  = \frac{\mathrm{Tr}[P_i\,S^2]
    - 2\,\mathrm{Tr}[P_i^2\,S]
    - \mathrm{Tr}[P_i\,Q]
    + 2\,\mathrm{Tr}[P_i^3]}
    {(T-1)(T-2)},
\end{equation}
and $\mathrm{SE}(\hat{p}_3)
= 3\,\mathrm{std}(G_3)/\sqrt{T}$.

\subsection{Computational complexity}
\label{app:complexity}

The bottleneck for both estimators is the $O(T\,N^4)$
accumulation of~$S$~\eqref{eq:S_accumulate}: for each of the
$T$ snapshots, the outer product $A_i \otimes B_i^T$ has
$N^4$ entries.  The subsequent matrix operations
($\mathrm{Tr}[S^2]$, $\mathrm{Tr}[S^3]$) are $O(N^6)$,
subdominant for $T \gg N^2$.

For the sample sizes and cutoffs in this paper
($T \leq 5 \times 10^4$, $N \leq 10$), the total
computation takes seconds on a single CPU core.  No
combinatorial enumeration of $\binom{T}{2}$ or
$\binom{T}{3}$ tuples is needed: the accumulator approach~\eqref{eq:p2_accumulate}
and~\eqref{eq:p3_accumulate} reduces the U-statistic
evaluation to a fixed number of matrix operations on the
accumulated sums~$S$ and~$Q$.

\section{Truncation Validity Proof}
\label{app:truncation}

We prove that the truncated linear and quadratic witnesses
\begin{align}
W_{\text{lin}}^{(N)} &= p_3^{(N)} - \frac{3t_N p_2^{(N)} - t_N^3}{2},
\label{eq:Wlin_trunc_app}\\
W_{\text{quad}}^{(N)} &= t_N p_3^{(N)} - \bigl(p_2^{(N)}\bigr)^2,
\label{eq:Wquad_trunc_app}
\end{align}
are valid entanglement tests for the original infinite-dimensional
state~$\rho$, at every Fock cutoff~$N \geq 1$.  Throughout,
$\Pi_N = \sum_{n=0}^{N} \ket{n}\!\bra{n}$ denotes the per-mode Fock
projector, $\rho_N = (\Pi_N \otimes \Pi_N)\rho(\Pi_N \otimes
\Pi_N)$ the projected state with
$t_N = \mathrm{tr}[\rho_N] \leq 1$, and
$p_k^{(N)} = \mathrm{tr}[(\rho_N^{T_B})^k]$ the projected
PT-moments.

\emph{Proposition.---}%
For any bipartite state~$\rho$ and any Fock cutoff~$N \geq 1$,
\begin{align}\label{eq:trunc_chain_app}
  W_{\text{lin/quad}}^{(N)} < 0
  \Longrightarrow
  \rho_N^{T_B} \not\geq 0
  &\Longrightarrow
  \rho^{T_B} \not\geq 0 \notag \\
  &\quad\Longrightarrow
  \rho \text{ is entangled.}
\end{align}
In particular, neither truncated witness produces false positives.
This is a structural property of the exact witness value; the
finite-sample estimator requires statistical thresholding
(Sec.~\ref{sec:protocol}).

\begin{proof}
We prove the contrapositive: if~$\rho$ is separable, then
$W_{\text{lin}}^{(N)} \geq 0$ and $W_{\text{quad}}^{(N)} \geq 0$
for every~$N$.

\emph{Step~1 (Separability $\Rightarrow$ PPT).}---
By the Peres criterion~\cite{peresSeparabilityCriterionDensity1996},
separability of~$\rho$ implies that its partial transpose is positive
semidefinite:
\begin{equation}
  \rho^{T_B}
  = \sum_j q_j\rho_j^A \otimes (\rho_j^B)^T
  \geq 0.
\end{equation}

\emph{Step~2 (Compression commutes with $T_B$).}---
Since $\Pi_N$ is real-diagonal in the Fock basis, partial
transposition and projection commute:
\begin{align}
  \rho_N^{T_B}
  &= \bigl[(\Pi_N \otimes \Pi_N)\rho
    (\Pi_N \otimes \Pi_N)\bigr]^{T_B} \\
  &= (\Pi_N \otimes \Pi_N)\rho^{T_B}
    (\Pi_N \otimes \Pi_N).
\end{align}
Compression preserves positive semidefiniteness: for any
$\ket{v}$,
\begin{equation}
  \bra{v}(\Pi_N \otimes \Pi_N)A
  (\Pi_N \otimes \Pi_N)\ket{v}
  = \braket{\Pi v | A | \Pi v}
  \geq 0
\end{equation}
whenever $A \geq 0$.
Hence $\rho^{T_B} \geq 0$ implies $\rho_N^{T_B} \geq 0$.

\emph{Step~3 ($\rho_N^{T_B} \geq 0$ implies both witnesses
non-negative).}---
Set $A = \rho_N^{T_B}$ and let $\{\lambda_i\}$ be its eigenvalues,
with $\lambda_i \geq 0$ and
$t_N = \mathrm{tr}[A] = \sum_i \lambda_i \leq 1$.  Write
$p_k^{(N)} = \mathrm{tr}[A^k] = \sum_i \lambda_i^k$.
(If $t_N = 0$ then $\rho_N = 0$ and both witnesses vanish;
henceforth assume $t_N > 0$.)

\emph{Step~3a (Quadratic witness).}---
Apply the Cauchy--Schwarz inequality to the vectors
$u_i = \lambda_i^{1/2}$ and $v_i = \lambda_i^{3/2}$:
\begin{equation}\label{eq:CS_moment_app}
  \Bigl(\sum_i \lambda_i\Bigr)
  \Bigl(\sum_i \lambda_i^3\Bigr)
  \geq
  \Bigl(\sum_i \lambda_i^2\Bigr)^2,
\end{equation}
i.e.\ $t_N \cdot p_3^{(N)} \geq \bigl(p_2^{(N)}\bigr)^2$, which is
$W_{\text{quad}}^{(N)} \geq 0$.

\emph{Step~3b (Linear witness).}---
Newton's identity for the third power sum in terms of the elementary
symmetric polynomials $e_k$ of the spectrum $\{\lambda_i\}$ reads
\begin{equation}\label{eq:Wlin_newton}
  p_3^{(N)}
  = e_1p_2^{(N)} - e_1e_2 + 3e_3,
\end{equation}
with $e_1 = \sum_i \lambda_i = t_N$,
$e_2 = (e_1^2 - p_2^{(N)})/2 = (t_N^2 - p_2^{(N)})/2$, and
$e_3 = \sum_{i<j<k}\lambda_i\lambda_j\lambda_k$.  Substituting and
rearranging,
\begin{align}\label{eq:Wlin_e3}
  W_{\text{lin}}^{(N)} = p_3^{(N)} - \frac{3t_Np_2^{(N)} - t_N^3}{2} &= 3e_3 \\ &= 3\!\sum_{i<j<k}\lambda_i\lambda_j\lambda_k
  \geq 0,
\end{align}
where non-negativity is immediate from $\lambda_i \geq 0$.
Equivalently,
$2p_3^{(N)} - 3t_Np_2^{(N)} + t_N^3 = 6e_3(\rho_N^{T_B}) \geq 0$:
the linear witness equals, up to a positive prefactor, the third
elementary symmetric polynomial of the compressed PT spectrum.
\end{proof}

\paragraph*{Setting $t_N = 1$ is weaker but valid.---}%
Throughout the main text we evaluate the simplified witnesses
obtained by setting $t_N = 1$ in
\eqref{eq:Wlin_trunc_app}--\eqref{eq:Wquad_trunc_app},
\begin{equation}\label{eq:simplified_witnesses}
  \widehat W_{\text{lin}}^{(N)}
    = p_3^{(N)} - \frac{3p_2^{(N)} - 1}{2},
  \qquad
  \widehat W_{\text{quad}}^{(N)}
    = p_3^{(N)} - \bigl(p_2^{(N)}\bigr)^2.
\end{equation}
On any positive compressed spectrum,
\begin{align}
\widehat W_{\text{quad}}^{(N)} - W_{\text{quad}}^{(N)}
  &= (1 - t_N)p_3^{(N)} \geq 0,
  \label{eq:Wquad_simplification}\\
\widehat W_{\text{lin}}^{(N)} - W_{\text{lin}}^{(N)}
  &= \tfrac{1-t_N}{2}\bigl[1 + t_N + t_N^2 - 3p_2^{(N)}\bigr] \geq 0,
  \label{eq:Wlin_simplification}
\end{align}
the first using $p_3^{(N)} \geq 0$ on PSD spectra, the second using
$p_2^{(N)} \leq t_N^2$ (from
$\sum_i \lambda_i^2 \leq (\sum_i\lambda_i)^2$ for non-negative
$\lambda_i$), which gives
$1 + t_N + t_N^2 - 3p_2^{(N)} \geq 1 + t_N - 2t_N^2 = (1-t_N)(1+2t_N) \geq 0$.
Hence the simplified witnesses dominate their $t_N$-corrected
counterparts on PPT spectra: separable states still pass
($\widehat W \geq W^{(N)} \geq 0$), and any $\widehat W < 0$
implies $W^{(N)} \leq \widehat W < 0$, which certifies
$\rho_N^{T_B} \not\geq 0$ by the proposition.  Setting $t_N = 1$
therefore trades detection power (the gap is of order $1 - t_N$) for
variance reduction and statistical convenience; in our benchmarks
$N \geq 6$ gives $1 - t_N < 10^{-3}$.

\paragraph*{Moment monotonicity.---}
The second projected moment satisfies $p_2^{(N)} \leq p_2$ for any
state, separable or entangled.  This follows from the Hilbert--Schmidt
block decomposition $\norm{A}_F^2 = \norm{\Pi A \Pi}_F^2 +
2\norm{\Pi A Q}_F^2 + \norm{Q A Q}_F^2$, with $Q = I - \Pi$.  The
third moment satisfies $p_3^{(N)} \leq p_3$ whenever
$\rho^{T_B} \geq 0$, since $x \mapsto x^3$ is Schur-convex on
$\mathbb{R}_+$ and the compressed eigenvalues are majorized by the
original ones (Cauchy interlacing), but $p_3^{(N)}$ can move in
either direction for entangled states.

\paragraph*{Truncation error bounds.---}
Let $\delta_N = 1 - t_N$ denote the truncation error.
Since $\norm{\rho^{T_B}}_{\mathrm{op}} \leq 1$ for all bipartite states (Supplemental Material Sec.~SVI), the
moment deficits are controlled by~$\delta_N$:
\begin{align}
  p_2 - p_2^{(N)}
  &\leq 2\delta_N,
  \label{eq:p2_deficit_app} \\
  p_3 - p_3^{(N)}
  &\leq 3\sqrt{2\delta_N p_2} + 8\delta_N.
  \label{eq:p3_deficit_app}
\end{align}
Propagating to the linear witness gives
$|W_{\text{lin}} - W_{\text{lin}}^{(N)}| \leq 3\sqrt{2\delta_N p_2} + 11\delta_N$,
so a sufficient condition for sign preservation is
$\delta_N \lesssim [W_{\text{lin}}^{(N)}]^2 / (18p_2)$.
For the quadratic witness,
$|W_{\text{quad}} - W_{\text{quad}}^{(N)}| \leq 3\sqrt{2\delta_Np_2} + 10\delta_N$,
with an analogous sign-preservation threshold.
For the states considered in this
work, $N \geq 6$ gives $\delta_N < 10^{-3}$, making the truncation
correction negligible compared to the statistical uncertainty.

\section{Finite-Sample Guarantees}
\label{app:finite-sample-guarantees}
Under the sole assumption of finite mean photon number ($\bar{n} <\infty$), a Hoeffding/Serfling variance decomposition combined with a pointwise Hilbert--Schmidt bound on the single-shot snapshots (Supplemental Material Sec.~SV) yields a state-independent non-asymptotic sample complexity
\begin{equation}
\label{eq:sample-scaling}
    T \ge \frac{K_k C_\text{pf}^2 (N+1)^{14/3}}{\delta\varepsilon^2}, \qquad K_2=32, \quad K_3 = 108,
\end{equation}
for estimating each truncated moment $p_k^{(N)}$ $(k=2,3)$ for additive accuracy $\varepsilon$ with confidence $1-\delta$, where $C_\text{pf}
$ is the universal Gill--Guta pattern-function constant \cite{Gill2003InvitationQT}; equivalently $T=O\left((N+1)^{14/3}/(\delta \varepsilon^2)\right)$. For states whose partially transposed Fock-basis matrix elements decay geometrically $\lvert [\rho^{T_B}]_{(nk),(ml)}\rvert \leq C_\rho r^{(n+k+l+m)/2}$ with $r< 1$, e.g. for two-mode squeezed vacuum (TMSV), the scaling becomes independent of $N$ at fixed truncated-moment accuracy.  Both the state-independent scaling and the geometric-decay scaling are conservative finite-sample guarantees; empirically, good estimates of $p_k^{(N)}$ ($k=2,3$) take between $10^3-10^4$ homodyne samples. The per-moment scaling \eqref{eq:sample-scaling} yields a finite-sample certification rule for the linear witness directly.  The triangle inequality
\begin{equation}\label{eq:Wlin_triangle}
  |\hat W_\text{lin}^{(N)} - W_\text{lin}^{(N)}|
  \leq \tfrac{3}{2}|\hat p_2^{(N)} - p_2^{(N)}|
  + |\hat p_3^{(N)} - p_3^{(N)}|
\end{equation}
combined with a union bound over two per-moment confidence intervals (allocating confidence $\delta/2$ to each) shows that
\begin{equation}\label{eq:rigorous_rule}
  \hat W_\text{lin}^{(N)} + \varepsilon < 0
  \quad\Longrightarrow\quad
  W_\text{lin}^{(N)} < 0
  \quad\text{with prob.}\geq 1 - \delta,
\end{equation}
whenever
\begin{equation}\label{eq:rigorous_T}
  T \geq \frac{C_WC_\text{pf}^2(N+1)^{14/3}}{\delta\varepsilon^2},
\end{equation}
with $C_W = \bigl(\tfrac{3}{2}\sqrt{2 K_2}+\sqrt{2 K_3}\bigr)^2 \approx 713$.
Combined with the truncation proposition \eqref{eq:trunc_chain}, \eqref{eq:rigorous_rule} provides a non-asymptotic, distribution-free certificate of $\rho^{T_B}\not\geq 0$, and hence of entanglement with no central-limit-theorem assumption.  An analogous bound for the quadratic witness carries an additional $(p_2 + \hat p_2)$ factor in the $p_2^2$ propagation and is given in the Supplemental Material Sec.~SV.
\section{Negativity Lower Bound Proof}
\label{app:negativity}

The negativity $\mathcal{N}(\rho) = (\norm{\rho^{T_B}}_1 - 1)/2$
equals the total weight on the negative eigenvalues
of~$\rho^{T_B}$~\cite{VidalWerner2002}.  It is an entanglement
monotone, and the logarithmic negativity
$E_\mathcal{N} = \log_2(1 + 2\mathcal{N})$ bounds the entanglement
cost under PPT operations~\cite{Plenio2005}.

\emph{Theorem.---}\label{thm:negativity}%
Let~$\rho$ be a bipartite state with PT-moments $p_1 = 1$, $p_2$,
and~$p_3$ satisfying $p_3 < p_2^2$.  Then:

\emph{(i)} The negativity satisfies the cubic bound
\begin{equation}\label{eq:neg_cubic_app}
  \mathcal{N}(\rho) \;\geq\; \beta_*,
\end{equation}
where $\beta_*$ is the smallest positive real root of
\begin{equation}\label{eq:neg_cubic_poly_app}
  g(u) = u^3 + 2\,p_2\,u^2 + p_3\,u + (p_3 - p_2^2).
\end{equation}

\emph{(ii)} The negativity also satisfies the closed-form rational bound
\begin{equation}\label{eq:neg_rational_app}
  \mathcal{N}(\rho)
  \;\geq\;
  \frac{p_2^2 - p_3}{p_2 + p_3 + \tfrac{1}{4}}\,.
\end{equation}

\emph{(iii)} For pure states ($p_2 = 1$), the cubic factors as
$g(u) = (u+1)(u^2 + u + p_3 - 1)$ and the bound becomes
\begin{equation}\label{eq:neg_pure_app}
  \mathcal{N}_{\mathrm{lb}}
  = \frac{-1 + \sqrt{5 - 4\,p_3}}{2}.
\end{equation}

The full proof uses only the Cauchy--Schwarz inequality and the
spectral range of partial transposes.  It is self-contained and avoids
the compactness and KKT arguments that a finite-dimensional reduction
would require.

\begin{proof}
Let \(A=\rho^{T_B}\), and write its eigenvalues as
\begin{align}
\{\alpha_j\}_j \cup \{-\beta_k\}_k,
\qquad
\alpha_j\ge 0,\ \beta_k>0.
\end{align}
Since \(A\) is Hermitian and \(\operatorname{Tr}A=1\),
\begin{align}
\sum_j \alpha_j-\sum_k \beta_k = 1,
\qquad
\nu:=\mathcal N(\rho)=\sum_k \beta_k.
\end{align}
Set
\begin{align}
x:=\sum_j \alpha_j^2,\qquad
y:=\sum_k \beta_k^2,
\end{align}
\begin{align}
R_+:=\sum_j \alpha_j^3,\qquad
R_-:=\sum_k \beta_k^3.
\end{align}
Then
\begin{align}
x+y=p_2,
\qquad
R_+-R_-=p_3.
\end{align}

\smallskip
\noindent\emph{(i) Cubic bound.}
By Cauchy--Schwarz on the positive eigenvalues,
\begin{align}
R_+ \ge \frac{x^2}{1+\nu}.
\end{align}
Since each \(\beta_k\le \nu\),
\begin{align}
R_-=\sum_k \beta_k^3 \le \nu\sum_k \beta_k^2 = \nu y.
\end{align}
Hence
\begin{align}
p_3 \ge \frac{(p_2-y)^2}{1+\nu}-\nu y
      =: f_\nu(y),
\end{align}
where \(f_\nu\) is strictly decreasing on \([0,p_2]\).

\smallskip
\noindent\emph{Case 1: \(\nu^2\le p_2\).}
By Cauchy--Schwarz, \(y\le \nu^2\). Therefore
\begin{align}
p_3 \ge f_\nu(\nu^2)
    = \frac{(p_2-\nu^2)^2}{1+\nu}-\nu^3.
\end{align}
After rearranging,
\begin{align}
g(\nu):=\nu^3+2p_2\nu^2+p_3\nu+(p_3-p_2^2)\ge 0.
\end{align}

\smallskip
\noindent\emph{Case 2: \(\nu^2\ge p_2\).}
Since \(y\le p_2\) and \(f_\nu\) is decreasing,
\begin{align}
p_3 \ge f_\nu(p_2) = -\nu p_2.
\end{align}
Thus
\begin{align*}
g(\nu)
&\ge \nu^3+2p_2\nu^2-\nu^2p_2-\nu p_2-p_2^2 \\
&= (\nu+p_2)(\nu^2-p_2)\ge 0.
\end{align*}
So \(g(\nu)\ge 0\) in both cases. Since
\begin{align}
g(0)=p_3-p_2^2=-\delta<0,
\qquad
g(u)\to+\infty \quad (u\to\infty),
\end{align}
\(g\) has a smallest positive root \(\beta_*\), and therefore
\begin{align}
\nu\ge \beta_*.
\end{align}

\smallskip
\noindent\emph{(ii) Rational bound.}
We use the spectral-range fact\cite{pt_neg_eigvals}
\begin{align}
\operatorname{spec}(\rho^{T_B})\subseteq \left[-\tfrac12,1\right],
\end{align}
proved in the Supplemental Material Sec.~SIII. Hence
\begin{align}
\beta_k\le \tfrac12,
\qquad
y\le \frac{\nu}{2},
\qquad
R_-\le \frac{y}{2}.
\end{align}
Therefore
\begin{align}
p_3 \ge \frac{(p_2-y)^2}{1+\nu}-\frac{y}{2}
      =: h_\nu(y),
\end{align}
where \(h_\nu\) is strictly decreasing on \([0,p_2]\).

\smallskip
\noindent\emph{Case 1: \(\nu/2\le p_2\).}
Then \(y\le \nu/2\), so
\begin{align}
p_3 \ge h_\nu(\nu/2)
    = \frac{(p_2-\nu/2)^2}{1+\nu}-\frac{\nu}{4}.
\end{align}
Rearranging gives
\begin{align}
\nu\left(p_2+p_3+\tfrac14\right)\ge p_2^2-p_3.
\end{align}

\smallskip
\noindent\emph{Case 2: \(\nu/2\ge p_2\).}
Then \(\nu\ge 2p_2\). Also
\begin{align}
p_3=R_+-R_- \ge -R_- \ge -\frac{p_2}{2}.
\end{align}
Hence
\begin{align}
p_2+p_3+\tfrac14 \ge \frac{p_2}{2}+\frac14>0,
\end{align}
and therefore
\begin{align*}
\nu\left(p_2+p_3+\tfrac14\right)
&\ge 2p_2\left(\frac{p_2}{2}+\frac14\right) \\
&= p_2^2+\frac{p_2}{2}
 \ge p_2^2-p_3.
\end{align*}
Thus, in both cases,
\begin{align}
\nu \ge \frac{p_2^2-p_3}{p_2+p_3+\tfrac14}.
\end{align}

\smallskip
\noindent\emph{(iii) Specialization to \(p_2=1\).}
Setting \(p_2=1\) gives
\begin{align}
g(u)=(u+1)(u^2+u+p_3-1),
\end{align}
whose positive root is
\begin{align}
\frac{-1+\sqrt{5-4p_3}}{2}.
\end{align}
This is the claimed expression.
\end{proof}
The bounds above apply to the \emph{exact normalized}
moments $(p_1, p_2, p_3)$ of the full state~$\rho$.  In the
truncated protocol, the estimators target the projected
moments $p_k^{(N)} = \mathrm{Tr}[(\rho_N^{T_B})^k]$ of the
subnormalized state $\rho_N = \Pi_N\rho\Pi_N$, with
$t = \mathrm{Tr}[\rho_N] \leq 1$.
To obtain a valid lower bound on~$\mathcal{N}(\rho)$,
the correct route uses the \emph{unnormalized}
negativity $\mathcal{N}_*(\rho_N)
= (\|\rho_N^{T_B}\|_1 - t)/2$, which satisfies
$\mathcal{N}_*(\rho_N) \leq \mathcal{N}(\rho)$ by
compression monotonicity (Supplemental Material Sec.~SVI).
The truncated negativity bounds derived therein
are expressed directly in terms of the subnormalized
moments $(t_N, p_2^{(N)}, p_3^{(N)})$ and do not require
renormalization of the projected state.

\emph{Corollary.---}
Under the same conditions,
$E_\mathcal{N}(\rho) \geq \log_2(1 + 2\,\mathcal{N}_{\mathrm{lb}})$.

\section{Robustness Proofs}
\label{app:robustness}

\emph{Proposition (Loss validity).---}%
For any bipartite state~$\rho$, any single-mode loss channels
$\mathcal{L}_{\eta_A}$, $\mathcal{L}_{\eta_B}$ with transmissivities
$\eta_A, \eta_B \in (0,1]$, and any Fock cutoff~$N \geq 1$:
\begin{equation}\label{eq:loss_validity_app}
  W^{(N)}\!\bigl((\mathcal{L}_{\eta_A}
  \otimes \mathcal{L}_{\eta_B})(\rho)\bigr) < 0
  \;\Longrightarrow\;
  \rho \;\text{is entangled.}
\end{equation}

\begin{proof}
The local loss channel $\mathcal{L}_{\eta_A} \otimes
\mathcal{L}_{\eta_B}$ is a local operation (LOCC with trivial
classical communication) and therefore preserves
separability~\cite{peresSeparabilityCriterionDensity1996}.
By the truncation validity proposition
(Appendix~\ref{app:truncation}),
$W^{(N)} < 0$ implies that
$(\mathcal{L}_{\eta_A} \otimes \mathcal{L}_{\eta_B})(\rho)$
is entangled.  Since separability is preserved under local loss, the
original state~$\rho$ must be entangled.  The chain reads:
\begin{equation}
  W^{(N)} < 0
  \;\Longrightarrow\;
  \mathcal{L}_\eta(\rho)\;\text{entangled}
  \;\Longrightarrow\;
  \rho\;\text{entangled.}
\end{equation}
The result covers detector inefficiency, fiber loss, and spatial mode
mismatch simultaneously: whenever the witness fires, it certifies
entanglement of the original state, with no calibration of~$\eta$
required.  The negativity bound likewise transfers:
$\mathcal{N}(\rho) \geq
\mathcal{N}_{\mathrm{lb}}\!\bigl(p_2(\mathcal{L}_\eta(\rho)),\,
p_3(\mathcal{L}_\eta(\rho))\bigr)$, since~$\mathcal{N}$ is
non-increasing under LOCC.
\end{proof}

\paragraph*{Phase noise.---}
A constant offset~$\delta$ in the local-oscillator phase has exactly
zero effect on the estimators, because the U-statistic kernels depend
only on angle differences $\theta_i - \theta_j$, in which~$\delta$
cancels.  Stochastic jitter
$\xi^{(i)} \sim \mathcal{N}(0, \sigma_\theta^2)$ introduces a bias
\begin{equation}\label{eq:phase_bias_app}
  |\mathrm{Bias}(\hat{p}_2)|
  \;\leq\;
  \sigma_\theta^2\,J_2(\rho) + O(\sigma_\theta^4),
\end{equation}
where the angular inertia $J_2(\rho) = \sum_{n,m,k,l}
[(m - n)^2 + (l - k)^2]\,|(\rho^{T_B})_{(nk),(ml)}|^2$
satisfies $J_2 \leq 2(N - 1)^2 p_2$.  For typical experimental
parameters ($N = 6$, $\sigma_\theta = 20\;\mathrm{mrad}$), the bias
is $\lesssim 2 \times 10^{-3}\,p_2$, negligible compared to the
statistical uncertainty at $T \sim 10^3$--$10^4$ samples.

\paragraph*{Discrete measurement phases.---}
The continuous uniform phase distribution can be replaced by a
discrete grid of $K_{\theta}$ equally spaced settings
$\theta_j = -\pi/2 + \pi j / K_{\theta}$ ($j = 0, \ldots, K_{\theta}-1$) per mode.
Provided $K_{\theta} \geq 2N + 1$, the resulting estimators for all
$p_k$ are exactly unbiased.  This follows from the fact that the
pattern-function products involve trigonometric polynomials of degree
at most~$2N$, which are integrated exactly by $K_{\theta} \geq 2N + 1$
equally spaced nodes (the Nyquist condition for angular harmonics of order up to $2N$; discrete orthogonality of complex
exponentials).  For cutoffs $N = 5$--$8$, this requires $K_{\theta} = 11$--$17$
phase settings per mode, or $K_{\theta} \times K_{\theta} = 121$--$289$ phase
combinations for the bipartite measurement.

\paragraph*{Ideal versus realistic pattern functions.---}
The robustness analysis above uses ideal pattern functions
($\eta = 1$ in the estimator).  Optical loss transforms the
input state $\rho \to \mathcal{L}_\eta(\rho)$ before
detection; the witness evaluates $W^{(N)}$ of this lossy
state.  The efficiency-dependent ``realistic'' pattern
functions of Ref.~\cite{richterRealisticPatternFunctions2000}
would reconstruct the \emph{pre-loss} state's moments but
require $\eta > 1/2$ and are unnecessary for the certification
guarantee.

\subsection{Robustness to loss and phase jitter}
\label{sec:app_loss}
\begin{figure*}
    \centering
    \includegraphics[width=0.9\linewidth]{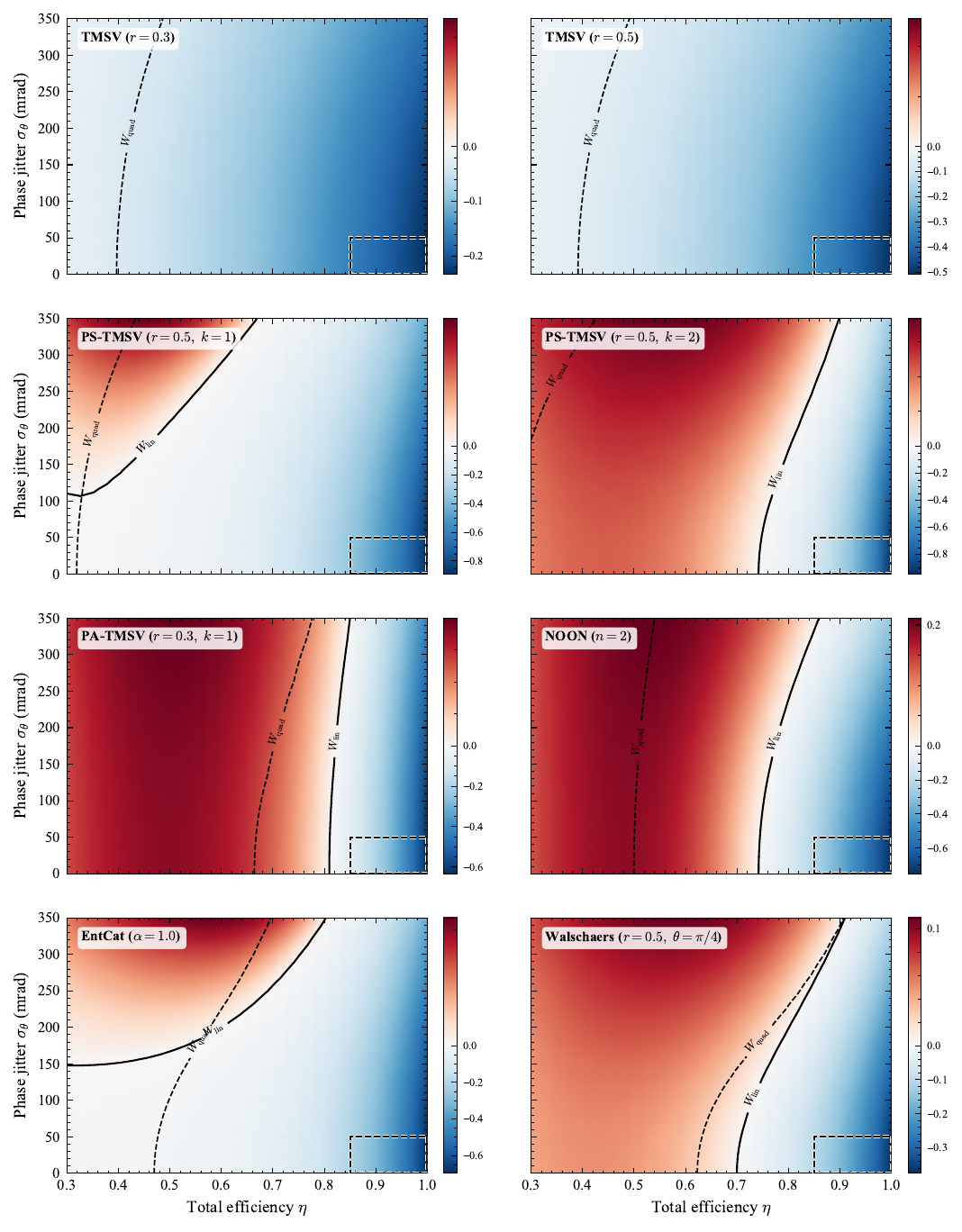}
    \caption{Witness value of the linear witness $W_\text{lin}$, for different states under detector ineffeciency and homodyne phase jitter. The solid (dashed) black line shows the boundary of entanglement detection for the linear (quadratic) witness.}
    \label{fig:app_state_survey}
\end{figure*}
\clearpage
\bibliography{apssamp}

\clearpage
\onecolumngrid

\setcounter{equation}{0}
\setcounter{figure}{0}
\setcounter{table}{0}
\setcounter{section}{0}
\setcounter{thm}{0}
\renewcommand{\theequation}{S\arabic{equation}}
\renewcommand{\thefigure}{S\arabic{figure}}
\renewcommand{\thetable}{S\arabic{table}}
\renewcommand{\thesection}{S\Roman{section}}

\begin{center}
  {\large\bfseries Supplemental Material:\\[2pt]
  Detecting non-Gaussian entanglement from single-copy homodyne measurements}
\end{center}
\medskip

This Supplemental Material provides the technical material supporting
the main text.  Section~\ref{sec:pattern_functions} summarizes the
numerical computation of the homodyne pattern functions.
Section~\ref{app:setup} collects notation: shadow snapshots,
U-statistic estimators, and the linear and quadratic witnesses.
Section~\ref{app:spectral} proves the spectral-range bound
$\operatorname{spec}(\rho^{T_B}) \subseteq [-\tfrac{1}{2}, 1]$ used by
the truncation and negativity arguments of the main text.
Sections~\ref{app:variance} and~\ref{app:gill_guta} derive the
variance bounds and the state-independent sample complexity quoted in
the main text, together with the finite-sample certification rules.
Finally, Sec.~\ref{app:trunc_neg} derives the negativity lower
bounds directly from the subnormalized truncated moments.

\section{Calculating Pattern Functions}
\label{sec:pattern_functions}
In this section, we present the numerical method that we use to calculate the pattern functions  $F_{n,m}(x_i, \theta_i)$. For the ideal, noiseless case, we implement the numerical method of Ref~\cite{Leonhardt1997-bs} and present it here for completeness. For the case where the detectors have efficiency $\eta <1$, we implement the realistic pattern functions introduced in Ref.~\cite{richterRealisticPatternFunctions2000}. Pattern functions allow for the reconstruction of elements of the density matrix in Fock basis $\rho_{n,m}$ from measurements of rotated quadrature operators $\hat{x}_\theta = \frac{1}{\sqrt{2}} \left( \hat{a}e^{-i\theta} + \hat{a}^\dagger e^{i\theta}\right)$~\cite{Leonhardt1997-bs, lvovskyContinuousvariableOpticalQuantumstate2009}:
\begin{equation}
    \rho_{m,n} = \frac{1}{2\pi}\int_{-\pi}^{+\pi} d\theta\int_{-\infty}^{+\infty}d{x_\theta}\bra{x_\theta}\rho\ket{x_\theta}F_{n,m}(x, \theta).
\end{equation}
Pattern functions $F_{n,m}(x,\theta)$ can be written as $F_{n,m}(x,\theta) = f_{n,m}(x) \exp(i(n-m)\theta)$ where $f_{n,m}(x)$ is the amplitude. The amplitudes have parity $f_{n,m}(-x) = (-1)^{n+m} f_{n,m}(x)$, so $F_{n,m}(-x, \theta+\pi) = (-1)^{n+m}\,e^{i(n-m)\pi}\,F_{n,m}(x,\theta) = F_{n,m}(x,\theta)$, while the quadrature density satisfies $\bra{-x_{\theta+\pi}}\rho\ket{-x_{\theta+\pi}} = \bra{x_\theta}\rho\ket{x_\theta}$. The integrand is thus invariant under $(x, \theta) \to (-x, \theta+\pi)$, and sampling phases uniformly from the half-interval $[-\pi/2, \pi/2)$, as in the main-text protocol, reproduces the full average.
\subsection{Ideal pattern functions}
 For completeness, we summarize the numerical recipe of Ref.~\cite{Leonhardt1997-bs}, which evaluates $F_{n,m}(x, \theta)$ through stable recurrence relations for the harmonic-oscillator wave functions.

 The amplitude $f_{n,m}(x)$ is obtained for $n \geq m$ from
\begin{eqnarray}
    f_{n,m}(x) = 2x\psi_m(x)\varphi_n(x) -\sqrt{2(m+1)} \psi_{m+1}(x) \varphi_n(x) - \sqrt{2(n+1)}\psi_m(x)\varphi_{n+1}(x).
\end{eqnarray}
 Here $\psi_m(x)$ are the regular harmonic-oscillator wave functions and $\varphi_n(x)$ the irregular (non-normalizable) solutions of the same Schr\"odinger equation; both families satisfy stable recurrence relations. For $\psi_m(x)$, Ref.~\cite{Leonhardt1997-bs} recommends using the  forward recurrence relation
\begin{equation}
    \psi_m(x) = \frac{1}{\sqrt{m}}\left[ \sqrt{2}x\psi_{m-1}(x) - \sqrt{m-1}\psi_{m-2}(x) \right],
\end{equation}
with the initial values
\begin{align}
    \psi_0(x) &= \pi^{-\frac{1}{4}} \exp\left(-\frac{x^2}{2}\right), \\
    \psi_1(x) &= \pi^{-\frac{1}{4}} \sqrt{2} x \exp\left(-\frac{x^2}{2}\right).
\end{align}
For determining $\varphi_m(x)$, we consider a semi-classical approximation. For $\lvert x\rvert < r_{4N} - (2r_{4N})^{-1/3}$, where $r_n = \sqrt{2n+1}$ is the Bohr-Sommerfeld radius, we use the backward relation
\begin{equation}
    \varphi_n(x) = \frac{1}{\sqrt{n+1}} \left[ \sqrt{2}x\varphi_{n+1}(x) - \sqrt{n+2}\varphi_{n+2}(x) \right],
\end{equation}
with initial values ($n=4N, n=4N-1$) given by
\begin{equation}
    \varphi_n(x) = \left(\frac{2\pi}{r_n \sin(t_n(x))}\right)^{1/2} \sin\left(\frac{r_n^2}{4}\left[\sin(2t_n(x)) - 2t_n(x) \right] + \frac{\pi}{4}\right),
\end{equation} where $t_n(x) = \arccos(x/r_n)$. For $\lvert x\rvert \ge r_{4N} - (2r_{4N})^{-1/3} $, we use the asymptotic relation
\begin{align}
    \varphi_n(x) &= \left(\frac{n}{2}\right)^{1/2} x^{-1}\varphi_{n-1}(x), \\
    \varphi_0(x) &= \pi^{1/4} x^{-1}\exp\left(-\frac{x^2}{2}\right).
\end{align}

The amplitudes are symmetric for $n,m$: $f_{n,m}(x) = f_{m,n}(x)$.
\begin{figure}
    \centering
    \includegraphics[width=0.6\linewidth]{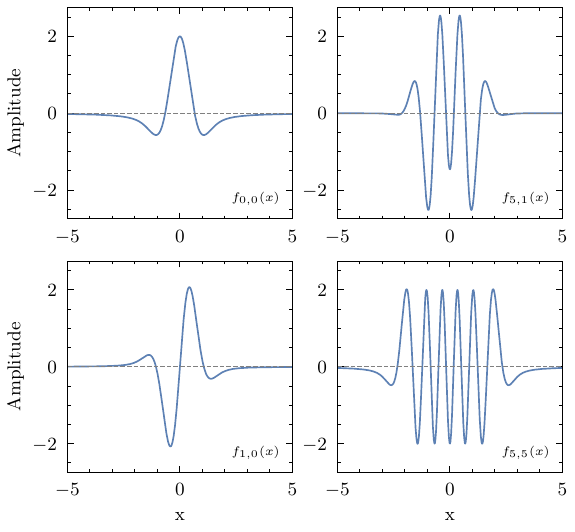}
    \caption{Amplitudes $f_{n,m}(x)$ for selected pattern functions. Units for $x$ are in natural units ($\hbar=1$).}
    \label{fig:idealpatternfunctions}
\end{figure}
\subsection{Realistic pattern functions}
If the detector efficiency is $\eta <1$, the \emph{ideal} pattern functions $F_{n,m}(x, \theta)$ are no longer unbiased. Instead, we need to use pattern functions that depend on the detector efficiency $F_{n,m}(x, \theta;\eta)$. The main result of Ref.\cite{richterRealisticPatternFunctions2000} shows how to express $F_{n,m}(x, \theta;\eta)$ as a finite sum of ideal pattern functions $F_{n,m}(x, \theta)$:
\begin{eqnarray}
    F_{n,m}(x, \theta;\eta) = \sqrt{\frac{\eta}{2\eta-1}}^{n+m+2}\sum_{k=0}^n \sqrt{\binom{n}{k}\binom{m}{k}}\left( 1 - \frac{1}{\eta} \right)^k F_{n-k,m-k}\left(\frac{x}{\sqrt{2\eta-1}}, \theta\right).
\end{eqnarray}
\begin{figure}
    \centering
    \includegraphics[width=0.6\linewidth]{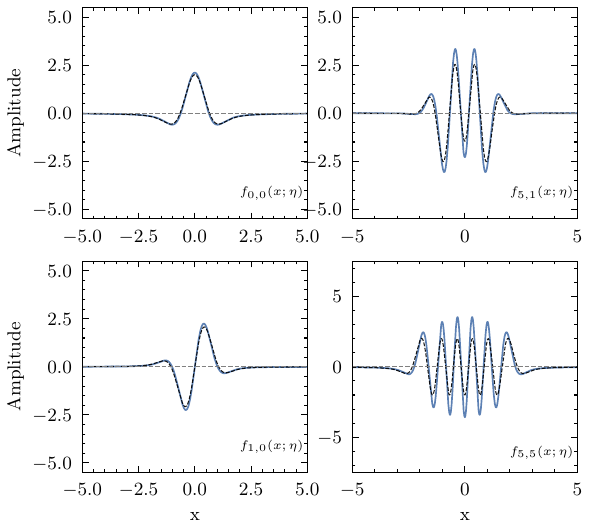}
    \caption{Amplitudes of realistic pattern functions $f_{n,m}(x; \eta)$ for selected $n,m$ and $\eta=0.95$. The dashed lines show the ideal pattern functions for the same $n,m$. Units for $x$ are in natural units ($\hbar=1$).}
    \label{fig:realisticpatternfunctions}
\end{figure}
\section{Notation and standing definitions}
\label{app:setup}

For ease of reference we collect the notation used throughout the
Supplemental Material.  All objects are defined in the main text, to
which we refer for their derivation; only the definitions needed by the
proofs below are repeated here.  We write $N$ for the largest Fock index
retained per mode (local dimension $d = N+1$),
$\Pi_N = \sum_{n=0}^N \ket{n}\!\bra{n}$ for the per-mode Fock projector,
and $\rho_N = (\Pi_N \otimes \Pi_N)\,\rho\,(\Pi_N \otimes \Pi_N)$ for the
projected, subnormalized state, with $t_N = \mathrm{Tr}[\rho_N] \leq 1$.
The estimators target the partial-transpose moments of this truncated
state, $p_k^{(N)} = \mathrm{Tr}[(\rho_N^{T_B})^k]$; where no confusion
arises we drop the superscript and write $p_k$, $\hat{p}_k$.

\emph{Shadow snapshots.}---Run~$i$ records the phases and quadratures
$X_i = (\theta_A^{(i)}, \theta_B^{(i)}, x_A^{(i)}, x_B^{(i)})$, from which
each mode is mapped to a Hermitian shadow matrix on the truncated Fock
space,
\begin{equation}\label{eq:shadow_matrix}
  \rho_M^{(i)}
  = \sum_{n,m=0}^{N}
    F_{n,m}\!\bigl(x_M^{(i)}, \theta_M^{(i)}\bigr)\,
    \ket{n}\!\bra{m},
  \qquad M \in \{A, B\},
\end{equation}
with $F_{n,m}(x,\theta) = f_{n,m}(x)\,e^{i(n-m)\theta}$ the homodyne
pattern functions of Sec.~\ref{sec:pattern_functions}.  We abbreviate
$A_i := \rho_A^{(i)}$ and $B_i := \rho_B^{(i)}$.  The snapshots are
unbiased for the truncated state,
$\mathbb{E}[\rho_A^{(i)} \otimes \rho_B^{(i)}] = \rho_N$, so the
partial-transposed bipartite snapshot
$\hat{\Sigma}_N^{T_B}(X_i) = A_i \otimes B_i^{T}$ has mean
$\rho_N^{T_B}$.

\emph{Kernels and estimators.}---The order-$2$ and order-$3$ U-statistic
kernels are the two- and three-snapshot traces
\begin{align}
  h_2(X_i, X_j)
  &= \mathrm{Tr}\!\bigl[\hat{\Sigma}_N^{T_B}(X_i)\,
     \hat{\Sigma}_N^{T_B}(X_j)\bigr],
  \label{eq:h2_def_sm}\\
  h_3(X_i, X_j, X_k)
  &= \mathrm{Tr}[A_i A_j A_k]\;
     \mathrm{Tr}[B_i^{T} B_j^{T} B_k^{T}],
  \label{eq:h3_fock_sm}
\end{align}
the factorized form of $h_3$ following from the decoupling of the forward
shift on~$A$ and the reverse shift on~$B$ established in the kernel
appendix of the main text.  With the symmetrizations $\tilde{h}_2 = h_2$
and $\tilde{h}_3 = \operatorname{Re} h_3$, the unbiased estimators from
$T$ runs are
\begin{equation}\label{eq:ustats_sm}
  \hat{p}_2^{(N)}
  = \binom{T}{2}^{-1}\!\sum_{i<j}\tilde{h}_2(X_i, X_j),
  \qquad
  \hat{p}_3^{(N)}
  = \binom{T}{3}^{-1}\!\sum_{i<j<k}\tilde{h}_3(X_i, X_j, X_k).
\end{equation}

\emph{Witnesses.}---The linear (detection) and quadratic
(negativity-bounding) witnesses are $W_{\text{lin}} = p_3 - (3p_2 - 1)/2$
and $W_{\text{quad}} = p_3 - p_2^2$, with $t_N$-corrected truncated forms
\begin{equation}\label{eq:witnesses_trunc_sm}
  W_{\text{lin}}^{(N)}
  = p_3^{(N)} - \tfrac{1}{2}\bigl(3 t_N\, p_2^{(N)} - t_N^3\bigr),
  \qquad
  W_{\text{quad}}^{(N)}
  = t_N\,p_3^{(N)} - \bigl(p_2^{(N)}\bigr)^2.
\end{equation}
The main text adopts the weaker but valid $t_N = 1$ simplification
throughout.  Being linear in the moments, $\hat{W}_{\text{lin}}$ is
exactly unbiased; the quadratic estimator carries the finite-sample bias
corrected in Sec.~\ref{app:bias_correction}.


\section{Spectral range of partial transposes}
\label{app:spectral}

The negativity bounds of the main text and the truncated-moment
bounds of Sec.~\ref{app:trunc_neg} use the fact that the spectrum of
a partial transpose is confined to $[-\tfrac{1}{2}, 1]$.  This
result is due to Rana~\cite{pt_neg_eigvals} for finite-dimensional states. Because our protocol operatoes on CV (infinite-dimensional) states and on the subnormalized projected state $\rho_N$, we give a short
self-contained proof valid in any (finite or infinite) dimension; it additionally yields the subnormalized form $\operatorname{spec}(\tau^{T_B}) \subseteq [-t/2,\, t]$ used by those truncated-moment bounds.

\begin{proposition}[Spectral range]\label{prop:opnorm}
For any bipartite state~$\rho$,
\begin{equation}\label{eq:opnorm_range}
  -\tfrac{1}{2}\,I \;\leq\; \rho^{T_B} \;\leq\; I,
\end{equation}
i.e.\ $\operatorname{spec}(\rho^{T_B}) \subseteq
[-\tfrac{1}{2},\, 1]$.  In particular
$\norm{\rho^{T_B}}_{\mathrm{op}} \leq 1$.
\end{proposition}

\begin{proof}
\emph{Step 1: pure states.}
Let $\ket{\psi} = \sum_i \sqrt{\lambda_i}\,\ket{a_i}\ket{b_i}$ be a
Schmidt decomposition, with $\lambda_i \geq 0$,
$\sum_i \lambda_i = 1$.  Taking the transpose in the Schmidt basis
$\{\ket{b_i}\}$ of mode~$B$,
\begin{equation}
  \bigl(\ket{\psi}\!\bra{\psi}\bigr)^{T_B}
  = \sum_{i,j} \sqrt{\lambda_i \lambda_j}\,
    \ket{a_i}\!\bra{a_j} \otimes \ket{b_j}\!\bra{b_i}.
\end{equation}
This operator is block-diagonal with respect to the decomposition
into the lines $\mathrm{span}\{\ket{a_i}\ket{b_i}\}$ and the planes
$\mathrm{span}\{\ket{a_i}\ket{b_j}, \ket{a_j}\ket{b_i}\}$ ($i < j$),
and annihilates the orthogonal complement.  On the lines it acts as
$\lambda_i$; on the planes it acts as
$\sqrt{\lambda_i\lambda_j}\,\sigma_x$, with eigenvalues
$\pm\sqrt{\lambda_i\lambda_j}$.  Hence
\begin{equation}
  \operatorname{spec}\bigl[(\ket{\psi}\!\bra{\psi})^{T_B}\bigr]
  \subseteq \{0\} \cup \{\lambda_i\}_i \cup
  \bigl\{\pm\sqrt{\lambda_i\lambda_j}\bigr\}_{i<j}.
\end{equation}
By the arithmetic--geometric mean inequality,
$\sqrt{\lambda_i\lambda_j} \leq (\lambda_i + \lambda_j)/2 \leq
\tfrac{1}{2}$ for $i \neq j$, while $\lambda_i \leq 1$.  This
establishes~\eqref{eq:opnorm_range} for pure states.

\emph{Step 2: mixed states.}
Write $\rho = \sum_k q_k \ket{\psi_k}\!\bra{\psi_k}$ with
$q_k \geq 0$, $\sum_k q_k = 1$.  Partial transposition is linear,
so $\rho^{T_B} = \sum_k q_k (\ket{\psi_k}\!\bra{\psi_k})^{T_B}$.
Each term satisfies the operator
inequalities~\eqref{eq:opnorm_range} by Step~1, and operator
inequalities are preserved under convex combinations.  Hence
$-\tfrac{1}{2} I \leq \rho^{T_B} \leq I$.

In infinite dimension, $\rho$ is trace class, the Schmidt sum in
Step~1 is countable, and $\rho^{T_B}$ is a Hilbert--Schmidt (hence
bounded, self-adjoint) operator, since partial transposition
preserves the Hilbert--Schmidt norm. The quadratic-form
inequalities above then hold on all of the Hilbert space, so the
conclusion is unchanged.
\end{proof}

Two consequences used elsewhere:
(i)~the negativity-bound proof in the main text uses
$\operatorname{spec}(\rho^{T_B}) \subseteq [-\tfrac{1}{2}, 1]$
(every negative eigenvalue satisfies $\beta_k \leq \tfrac{1}{2}$);
(ii)~the truncation-error bounds in the main text use
$\norm{\rho^{T_B}}_{\mathrm{op}} \leq 1$.  The subnormalized
version $\mathrm{Tr}[\tau] = t \Rightarrow\operatorname{spec}(\tau^{T_B}) \subseteq [-t/2,\, t]$ for $\tau \geq 0$follows by applying Proposition~\ref{prop:opnorm} to $\tau/t$ and is stated as a lemma in Sec.~\ref{app:spectral_range}.

\section{Hoeffding variance decomposition and sample complexity bounds}
\label{app:variance}

We derive the variance of the U-statistic estimators $\hat{p}_2$
and~$\hat{p}_3$ via the Hoeffding decomposition.  Together with the
pointwise Hilbert--Schmidt bound on the single-shot snapshots
(Sec.~\ref{app:gill_guta}), this fixes both the state-independent
sample complexity quoted in the main text and the sharper,
$N$-independent rate for states with geometrically decaying Fock
support.  The bias correction of the quadratic witness and the
jackknife standard errors used in the figures follow at the end.

Throughout, $N$ denotes the highest retained Fock number (local
dimension $d = N+1$), all index sums run from $0$ to~$N$, and the
estimators target the truncated state.  To keep the notation light we
write $\rho$ for $\rho_N$ and $p_k$ for $p_k^{(N)}$, dropping the
truncation superscript (cf.\ Sec.~\ref{app:setup}).

\subsection{Hoeffding decomposition}

The variance of an order-$k$ U-statistic with symmetric
kernel~$\tilde{h}_k$ admits the exact
decomposition~\cite{Serfling1980,Lee1990}
\begin{equation}\label{eq:hoeffding_general}
  \mathrm{Var}(\hat{p}_k)
  = \sum_{c=1}^{k} \binom{k}{c}^2
    \frac{\sigma_{k,c}^2}{\binom{T}{c}},
\end{equation}
where the Hoeffding components are
\begin{equation}
  \sigma_{k,c}^2
  = \mathrm{Var}\!\bigl(
    \mathbb{E}[\tilde{h}_k(X_1,\ldots,X_k)
    \mid X_1,\ldots,X_c]
    \bigr)
\end{equation}
(the variance of the conditional expectation given $c$~of
the $k$~arguments, averaged over the remaining $k - c$).

For $\hat{p}_2$ (order $k = 2$):
\begin{equation}\label{eq:hoeffding_p2}
  \mathrm{Var}(\hat{p}_2)
  = \frac{4(T\!-\!2)}{T(T\!-\!1)}\,\sigma_{2,1}^2
  + \frac{2}{T(T\!-\!1)}\,\sigma_{2,2}^2
  = \frac{4\sigma_{2,1}^2}{T} + O(T^{-2}).
\end{equation}

For $\hat{p}_3$ (order $k = 3$):
\begin{equation}\label{eq:hoeffding_p3}
  \mathrm{Var}(\hat{p}_3)
  = \sum_{c=1}^{3} \binom{3}{c}^2
    \frac{\sigma_{3,c}^2}{\binom{T}{c}}
  = \frac{9\sigma_{3,1}^2}{T} + O(T^{-2}).
\end{equation}

The sample complexity for estimating $p_k$ to additive accuracy
$\varepsilon$ at confidence $1 - \delta$ is therefore determined
(via Chebyshev's inequality) by the leading component
$\sigma_{k,1}^2$.  We bound this quantity in the following
subsections.

\subsection{First Hoeffding projections}

The leading variance components are determined by the
\emph{first Hoeffding projections}:
\begin{align}
  G_2(X) &= \mathbb{E}_{X'}[h_2(X, X') \mid X] - p_2
  = \mathrm{Tr}\!\bigl[\hat{\Sigma}_N^{T_B}(X)\,
    \rho^{T_B}\bigr] - p_2,
  \label{eq:G2} \\
  G_3(X) &= \mathbb{E}_{X',X''}[h_3(X, X', X'') \mid X] - p_3
  = \mathrm{Tr}\!\bigl[\hat{\Sigma}_N^{T_B}(X)\,
    (\rho^{T_B})^2\bigr] - p_3,
  \label{eq:G3}
\end{align}
where $\hat{\Sigma}_N^{T_B}(X) = \rho_A^{(X)} \otimes
(\rho_B^{(X)})^T$ is the partial-transposed single-shot shadow.

\emph{Proof of~\eqref{eq:G2} and~\eqref{eq:G3}.}
The order-two kernel is the Hilbert--Schmidt overlap of two snapshots,
\begin{equation}
  h_2(X,X')
  = \mathrm{Tr}\bigl[\hat{\Sigma}_N^{T_B}(X)\,
    \hat{\Sigma}_N^{T_B}(X')\bigr] .
\end{equation}
Fix the first argument at $X$ and average over the second.  The
snapshots are i.i.d.\ and unbiased,
$\mathbb{E}[\hat{\Sigma}_N^{T_B}] = \rho^{T_B}$, so linearity of the
trace lets us replace the second snapshot by its mean,
\begin{equation}
  \mathbb{E}_{X'}\!\left[h_2 \mid X\right]
  = \mathrm{Tr}\Bigl[\hat{\Sigma}_N^{T_B}(X)\,
    \mathbb{E}\bigl[\hat{\Sigma}_N^{T_B}(X')\bigr]\Bigr]
  = \mathrm{Tr}\bigl[\hat{\Sigma}_N^{T_B}(X)\,\rho^{T_B}\bigr] .
\end{equation}
Subtracting $p_2 = \mathrm{Tr}[(\rho^{T_B})^2]$ gives~\eqref{eq:G2}.

The order-three case works the same way.  Conditioning the kernel
$h_3 = \mathrm{Tr}[\hat{\Sigma}_N^{T_B}(X)\,\hat{\Sigma}_N^{T_B}(X')\,
\hat{\Sigma}_N^{T_B}(X'')]$ on $X$ and averaging over the two remaining
snapshots replaces each of them by $\rho^{T_B}$, so the two factors
collapse into a single power,
\begin{equation}
  \mathbb{E}_{X',X''}\!\left[h_3 \mid X\right]
  = \mathrm{Tr}\bigl[\hat{\Sigma}_N^{T_B}(X)\,(\rho^{T_B})^2\bigr] .
\end{equation}
Subtracting $p_3 = \mathrm{Tr}[(\rho^{T_B})^3]$
gives~\eqref{eq:G3}. \hfill$\square$

Equations~\eqref{eq:G2}--\eqref{eq:G3} show that each first
projection is a \emph{linear} functional of a single snapshot: $G_k(X)$
depends on the random shadow $\hat{\Sigma}_N^{T_B}(X)$ only through its
Hilbert--Schmidt overlap with the fixed operator $(\rho^{T_B})^{k-1}$.
It is convenient to use this linearity by vectorizing.  Stacking the
random snapshot into the vector
$\mathbf{z}(X) = \mathrm{vec}(\hat{\Sigma}_N^{T_B}(X))$ and the fixed,
state-dependent operator into the coefficient vector
$\mathbf{r}_k = \mathrm{vec}\bigl((\rho^{T_B})^{k-1}\bigr)$, the
trace--overlap identity
$\mathrm{Tr}[A^\dagger B] = \mathrm{vec}(A)^\dagger\,\mathrm{vec}(B)$
turns the trace in \eqref{eq:G2}--\eqref{eq:G3} into an ordinary
Euclidean inner product,
\begin{equation}\label{eq:Gk_linear}
  G_k(X) = \mathbf{r}_k^\dagger\,\mathbf{z}(X) - p_k .
\end{equation}
The only random object is now $\mathbf{z}(X)$, and $G_k$ is linear in
it.  Its variance is therefore a single quadratic form in the snapshot
second-moment matrix $\mathbf{C} = \mathbb{E}[\mathbf{z}\mathbf{z}^\dagger]$,
\begin{equation}\label{eq:sigma_bilinear}
  \sigma_{k,1}^2
  = \mathrm{Var}(G_k)
  = \mathbf{r}_k^\dagger \mathbf{C}\, \mathbf{r}_k
    - p_k^2,
  \qquad
  \|\mathbf{r}_k\|^2 = p_{2(k-1)} .
\end{equation}
This is the form used throughout the rest of the section.  It separates
the variance into two pieces that we bound by different means: the
\emph{state-independent} measurement fluctuations, carried entirely by
$\mathbf{C}$ and bounded entry by entry in Sec.~\ref{app:si_bound}, and
the \emph{state-dependent} weights $\mathbf{r}_k$, controlled through
the Fock-space decay of $\rho^{T_B}$ in Sec.~\ref{app:gd_bounds}.

\subsection{Entry-wise control of the snapshot second moment}
\label{app:si_bound}

The state-independent variance bound quoted in the main text is obtained
most directly by the pointwise Hilbert--Schmidt argument of
Sec.~\ref{app:gill_guta}.  What that argument does not provide, and what
the state-\emph{dependent} geometric-decay analysis of
Sec.~\ref{app:gd_bounds} below requires, is a bound on the snapshot
second-moment matrix
$\mathbf{C} = \mathbb{E}[\mathbf{z}\mathbf{z}^\dagger]$ \emph{entry by
entry}: geometric decay weights each Fock index of $\rho^{T_B}$
separately, so the trace of~$\mathbf{C}$ alone is too coarse.  We record
the entry-wise bound here.

\begin{lemma}[Entry-wise bound on $\mathbf{C}$]\label{lem:entrywise}
For any bipartite state~$\rho$ (possibly entangled) and
any Fock cutoff~$N$, the entries of the second-moment matrix
satisfy
\begin{equation}\label{eq:C_entry_bound}
  |C_{\alpha\beta}|
  \;\leq\;
  (\Gamma_{nm}^A)^{1/4}\,(\Gamma_{lk}^B)^{1/4}\,
  (\Gamma_{n'm'}^A)^{1/4}\,(\Gamma_{l'k'}^B)^{1/4},
\end{equation}
where $\alpha = ((n,k),(m,l))$, $\beta = ((n',k'),(m',l'))$
are composite Fock indices, and
$\Gamma_{nm}^M = \int f_{nm}(x)^4\,\mu_M(x)\,dx$
is the fourth moment of the pattern function
under the angle-averaged marginal quadrature distribution
$\mu_M(x) = \pi^{-1}\int_0^\pi p_M(x|\theta)\,d\theta$
of mode~$M$.
\end{lemma}

\begin{proof}
Written out in the single-mode pattern functions, the entry is an
expectation over the joint outcome $(X_A, X_B)$,
\begin{equation}
  C_{\alpha\beta}
  = \mathbb{E}\bigl[
    F_{nm}^A(X_A)\,F_{lk}^B(X_B)\,
    \overline{F_{n'm'}^A(X_A)}\,\overline{F_{l'k'}^B(X_B)}
    \bigr] .
\end{equation}
Cauchy--Schwarz on this expectation bounds the off-diagonal entry by
the two diagonal ones,
\begin{equation}\label{eq:CS_offdiag}
  |C_{\alpha\beta}|^2 \;\leq\; C_{\alpha\alpha}\,C_{\beta\beta} .
\end{equation}
A diagonal entry contains only squared moduli, so the measurement
phases cancel and it reduces to an average of pattern functions under
the marginal quadrature distribution,
\begin{equation}
  C_{\alpha\alpha}
  = \mathbb{E}\bigl[|F_{nm}^A|^2\,|F_{lk}^B|^2\bigr]
  = \mathbb{E}\bigl[f_{nm}^A(x_A)^2\,f_{lk}^B(x_B)^2\bigr] .
\end{equation}
Applying Cauchy--Schwarz once more, now to factor the two modes apart,
\begin{equation}
  C_{\alpha\alpha}
  \;\leq\; \sqrt{\mathbb{E}\bigl[f_{nm}^A(x_A)^4\bigr]}\,
           \sqrt{\mathbb{E}\bigl[f_{lk}^B(x_B)^4\bigr]}
  = \sqrt{\Gamma_{nm}^A}\,\sqrt{\Gamma_{lk}^B} ,
\end{equation}
where the last step is the definition of the fourth moment: a single
power of the fourth-power pattern function depends only on the marginal
distribution, $\mathbb{E}[f_{nm}^A(x_A)^4]
= \int f_{nm}^A(x)^4\,\mu_A(x)\,dx = \Gamma_{nm}^A$.  The identical
bound holds for $C_{\beta\beta}$ in terms of $\Gamma_{n'm'}^A$ and
$\Gamma_{l'k'}^B$.  Substituting both into~\eqref{eq:CS_offdiag} and
taking the square root yields~\eqref{eq:C_entry_bound}.
\end{proof}

The fourth moments inherit the Gill--Guta pattern-function supremum
estimate of Sec.~\ref{app:gill_guta} (Lemma~\ref{lem:gill_guta}).  Since
each $\mu_M$ is a probability measure,
$(\Gamma_{nm})^{1/4} \leq \|f_{nm}\|_\infty$, and because
$\|f_{nm}\|_\infty^2$ is a single non-negative term of the Gill--Guta
sum, each amplitude obeys the per-element bound
\begin{equation}\label{eq:fnm_sup_bound}
  \|f_{nm}\|_\infty
  \;\leq\; \sqrt{C_{\mathrm{pf}}}\,(\max(n,m)+1)^{7/6}.
\end{equation}
The entry-wise bound~\eqref{eq:C_entry_bound} together
with~\eqref{eq:fnm_sup_bound} is all that the geometric-decay analysis
below requires.

\subsection{State-dependent bounds under geometric decay}
\label{app:gd_bounds}

For states satisfying the geometric Fock-decay condition
\begin{equation}\label{eq:GD}
  |[\rho^{T_B}]_{(nk),(ml)}|
  \;\leq\; C_\rho\,r^{(n+k+m+l)/2},
  \quad r < 1,
  \tag{GD}
\end{equation}
(which holds for TMSV with $C_\rho = \mathrm{sech}^2\!r_s$
and $r = \tanh r_s$, and for photon-subtracted states with
the same~$r$ and modified~$C_\rho$), the leading variance
components are bounded independently of~$N$.

\begin{theorem}[$N$-independent bounds under geometric decay]
\label{thm:gd_bounds}
Under~\eqref{eq:GD}:
\begin{align}
  \sigma_{2,1}^2 &\leq C_\rho^2\,\tilde{S}^{\,4},
  \label{eq:bilinear_p2}\\
  \sigma_{3,1}^2 &\leq
    \frac{C_\rho^4}{(1\!-\!r)^4}\,\tilde{S}^{\,4},
  \label{eq:bilinear_p3}
\end{align}
where the convergent sum
\begin{equation}
  \tilde{S}
  = \sum_{n,m=0}^{\infty}
    r^{(n+m)/2}\,(\Gamma_{nm})^{1/4}
  < \infty
\end{equation}
involves the pattern function fourth moments
$\Gamma_{nm} = \int |f_{nm}(x)|^4\,\mu(x)\,dx$.
\end{theorem}

Since $\mu$ is a probability measure,
$(\Gamma_{nm})^{1/4}
= \bigl(\int |f_{nm}|^4\,\mu\,dx\bigr)^{1/4}
\leq \|f_{nm}\|_\infty$,
so $\tilde{S} \leq S_r
:= \sum_{n,m\geq 0} r^{(n+m)/2}\,\|f_{nm}\|_\infty$.
The Gill--Guta per-element
bound~\eqref{eq:fnm_sup_bound} gives
$\|f_{nm}\|_\infty
= O\!\bigl((\max(n,m)+1)^{7/6}\bigr)$,
so the geometric factor $r^{(n+m)/2}$ dominates the
polynomial growth and $S_r < \infty$.

\begin{proof}
Start from the quadratic form~\eqref{eq:sigma_bilinear} and write it in
components, with $R_\alpha = [\rho^{T_B}]_{(nk),(ml)}$ the entries of
$\mathbf{r}_2$.  Dropping the $-p_2^2$ term and applying the triangle
inequality,
\begin{equation}
  \sigma_{2,1}^2
  \;\leq\; \mathbb{E}[|G_2|^2]
  = \sum_{\alpha,\beta} \bar{R}_\alpha\,C_{\alpha\beta}\,R_\beta
  \;\leq\; \sum_{\alpha,\beta}
       |R_\alpha|\,|C_{\alpha\beta}|\,|R_\beta| .
\end{equation}
Insert the entry-wise bound~\eqref{eq:C_entry_bound} for
$|C_{\alpha\beta}|$ and the geometric decay~\eqref{eq:GD} for the
weights,
\begin{equation}
  |R_\alpha|
  = \bigl|[\rho^{T_B}]_{(nk),(ml)}\bigr|
  \;\leq\; C_\rho\,r^{(n+m)/2}\,r^{(k+l)/2} .
\end{equation}
Both bounds are products over the individual Fock indices, so the
eight-index sum separates into four identical two-index sums, one for
each mode on each side of $\mathbf{C}$,
\begin{equation}
  \sigma_{2,1}^2
  \;\leq\; C_\rho^2\,\bigl(\tilde{S}_A^{(N)}\bigr)^2
       \bigl(\tilde{S}_B^{(N)}\bigr)^2
  \;\leq\; C_\rho^2\,\tilde{S}^4 ,
  \qquad
  \tilde{S}_M^{(N)}
  = \sum_{n,m=0}^{N} r^{(n+m)/2}\,(\Gamma_{nm}^M)^{1/4}
  \;\leq\; \tilde{S} ,
\end{equation}
which is~\eqref{eq:bilinear_p2}.

The bound on $\sigma_{3,1}^2$ follows the same steps with
$\mathbf{r}_3 = \mathrm{vec}((\rho^{T_B})^2)$ in place of
$\mathbf{r}_2$.  The one new element is that squaring $\rho^{T_B}$ sums
over an internal Fock index; under~\eqref{eq:GD} that sum is geometric
and contributes a factor $(1-r)^{-2}$,
\begin{equation}
  \bigl|[(\rho^{T_B})^2]_{(ab),(cd)}\bigr|
  \;\leq\; \frac{C_\rho^2}{(1-r)^2}\,r^{(a+b+c+d)/2} .
\end{equation}
Carrying this through the same factorization replaces $C_\rho^2$ by
$C_\rho^4/(1-r)^4$ and gives~\eqref{eq:bilinear_p3}.
\end{proof}

Under~\eqref{eq:GD}, the sample complexity becomes
\begin{equation}
  T = O\!\left(\frac{C(\rho)}{\varepsilon^2}\right),
\end{equation}
independent of the Fock cutoff~$N$.

\subsection{Bias correction of the quadratic witness}
\label{app:bias_correction}

The linear witness estimator
$\hat{W}_{\text{lin}} = \hat{p}_3 - (3\hat{p}_2 - 1)/2$ is a linear
combination of unbiased U-statistics and is therefore exactly
unbiased.  The quadratic witness, in
contrast, involves the square of $\hat{p}_2$ and acquires a
finite-sample bias, which we quantify and correct here.

\subsubsection{Exact bias identity}

The quadratic witness estimator
$\hat{W}_{\text{quad}} = \hat{p}_3 - \hat{p}_2^2$ has bias
\begin{equation}\label{eq:bias_identity}
  \mathbb{E}[\hat{W}_{\text{quad}}] - W_{\text{quad}} = -\mathrm{Var}(\hat{p}_2)
  = -\frac{4\sigma_{2,1}^2}{T} + O(T^{-2}).
\end{equation}
This negative bias shifts the estimator toward more
negative values, which must be accounted for in any
finite-sample decision rule.  We describe two correction
modes: a practical jackknife-based correction used in the
main protocol, and a conservative deterministic upper
bound available as an alternative.

\subsubsection{Mode~1: practical correction (jackknife)}
\label{app:bias_practical}

We estimate $\mathrm{Var}(\hat{p}_2)$ from the first Hoeffding
projection.  Since $\hat{p}_2$ is an order-2 U-statistic, its leading
variance term is $\mathrm{Var}(\hat{p}_2) = 4\sigma_{2,1}^2/T +
O(T^{-2})$, with $\sigma_{2,1}^2 = \mathrm{Var}(G_2(X))$.

The jackknife estimates this variance directly from the data, with no
analytic input about the state~\cite{Arvesen1969}.  The idea is to ask how sensitive the estimate is to each individual snapshot: one
recomputes $\hat{p}_2$ with each snapshot deleted in turn, and a
snapshot whose removal barely moves $\hat{p}_2$ is counted as carrying
little information about the estimator's variability, while one that
moves it a lot counts for more.  For a U-statistic this leave-one-out
sensitivity is, to leading order, exactly the centered first Hoeffding
projection $h_1^{(2)}(X_i) - \hat{p}_2$ at the deleted snapshot, so the
variance is estimated by the empirical spread of these projections,
\begin{equation}\label{eq:B_JK}
  \hat{B}_{\mathrm{JK}}
  = \widehat{\mathrm{Var}}(\hat{p}_2)
  = \frac{4\,\hat{\sigma}_{2,1}^2}{T},
  \qquad
  \hat{\sigma}_{2,1}^2
  = \frac{1}{T-1}\sum_{i=1}^{T}
    \bigl(h_1^{(2)}(X_i) - \hat{p}_2\bigr)^2,
\end{equation}
where $h_1^{(2)}(X_i)$ is the first Hoeffding projection of the
$\hat{p}_2$ kernel evaluated at snapshot~$X_i$
(Sec.~\ref{app:variance}).  By the law of large numbers this empirical
spread converges to the true variance,
$\hat{B}_{\mathrm{JK}} \to \mathrm{Var}(\hat{p}_2)$ in probability as
$T \to \infty$.

The practical corrected witness is
\begin{equation}\label{eq:W_practical}
  \hat{W}'_{\text{quad}} = \hat{p}_3 - \hat{p}_2^2
  + \hat{B}_{\mathrm{JK}}.
\end{equation}
This correction compensates the leading finite-sample bias,
and the resulting one-sided test~(the asymptotic
studentized rule of the main text)
is asymptotically valid: as $T \to \infty$, the bias
correction is exact and the normal approximation is
justified by the U-statistic CLT \cite{Lee1990}.  Finite-sample
calibration is empirical, validated by the Monte Carlo
experiments in the demonstration section of the
main text.

\subsubsection{Mode~2: conservative correction
  (deterministic upper bound)}
\label{app:bias_conservative}

If a rigorous finite-sample guarantee is desired, the
jackknife estimate can be replaced by a deterministic upper
bound.  Using the state-independent variance
bound~\eqref{eq:var_p2_GG} of Sec.~\ref{app:gill_guta}:
\begin{equation}\label{eq:bias_bound}
  \mathcal{B}
  := \frac{16\,C_{\mathrm{pf}}^2\,(N+1)^{14/3}\,\hat{p}_2}{T}
  \;\geq\; \mathrm{Var}(\hat{p}_2)\quad\text{(in expectation)},
\end{equation}
where $\hat{p}_2$ replaces $p_2$ (since $p_2$ is unknown)
and $C_{\mathrm{pf}}$ is the Gill--Guta pattern-function constant
(Sec.~\ref{app:gill_guta}).
The bound $\mathcal{B}$ is itself random through its dependence on
$\hat{p}_2$, and the inequality holds in expectation because
$\mathbb{E}[\hat{p}_2] = p_2$.
For a fully deterministic bound, one may use $p_2 \leq 1$
in place of $\hat{p}_2$.

The conservative corrected witness
$\hat{W}'_{\mathrm{rig}} = \hat{p}_3 - \hat{p}_2^2
+ \mathcal{B}$ satisfies
$\mathbb{E}[\hat{W}'_{\mathrm{rig}}] \geq W_{\text{quad}}$ by
construction (positive bias in expectation).  Combined
with the Chebyshev-based rigorous certification rule for the
quadratic witness (Sec.~\ref{app:quad_rule}), this yields a
non-asymptotic finite-sample certification guarantee.
The conservative correction is typically looser than the
jackknife-based correction: $\mathcal{B}$ overestimates
$\mathrm{Var}(\hat{p}_2)$ by a constant factor that
grows with~$N$, reducing detection significance for a
given sample size.

Under the geometric decay condition~\eqref{eq:GD}, the
state-dependent bound
$\mathcal{B}_{\mathrm{GD}}
= 4\,C_\rho^2\,\tilde{S}^4 / T$ may be used instead,
giving a tighter conservative correction that is
independent of~$N$.

\subsection{Jackknife standard errors}
\label{app:jackknife_se}

The shaded uncertainty bands in the figures of the main text are
jackknife (influence-function) standard errors~\cite{Arvesen1969}.
They use the same leave-one-out variance estimate as the bias
correction of Sec.~\ref{app:bias_practical}, now applied to each
estimator in turn.  For the order-$k$ U-statistic $\hat{p}_k$, the
empirical first Hoeffding projection $\hat{G}_k(X_i)$ is evaluated for
every snapshot via the accumulator identities of the main-text kernel
appendix, and
\begin{equation}\label{eq:jackknife_se}
  \mathrm{SE}(\hat{p}_k)
  = \frac{k\,\widehat{\mathrm{std}}(\hat{G}_k)}{\sqrt{T}},
  \qquad k \in \{2, 3\},
\end{equation}
where $\widehat{\mathrm{std}}$ denotes the empirical standard
deviation over the $T$ snapshots.  For the witnesses, the plug-in
influence functions
\begin{equation}\label{eq:jackknife_se_witness}
  \hat\psi_{\text{lin}}(X_i) = 3\hat{G}_3(X_i) - 3\hat{G}_2(X_i),
  \qquad
  \hat\psi_{\text{quad}}(X_i) = 3\hat{G}_3(X_i)
  - 4\hat{p}_2\,\hat{G}_2(X_i),
\end{equation}
give
$\hat\sigma_{W} = \widehat{\mathrm{std}}(\hat\psi)/\sqrt{T}$,
which is the standard error entering the asymptotic studentized rule
of the main text.  These estimators are consistent by the U-statistic
law of large numbers; their finite-$T$ accuracy is validated against
Monte Carlo standard errors in the demonstration section of the main
text.


\section{Sample complexity via pattern-function supremum
  estimates}
\label{app:gill_guta}
Here we derive the sample complexity of the moment estimators and the linear witness. Estimating the moments and witnesses to additive accuracy~$\varepsilon$
at confidence $1-\delta$ costs
$T = O\!\bigl((N+1)^{14/3}/(\delta\varepsilon^2)\bigr)$ samples, with
the explicit constants $K_2 = 32$ and $K_3 = 108$ quoted in the main
text.

For the derivation, we deriva a worst-case bound on the
size of one homodyne snapshot, measured by its Hilbert--Schmidt norm.
This bound holds \emph{pointwise}, meaning uniformly over every
measurement outcome, and follows from the pattern-function supremum
estimate of Gill and Guta~\cite{Gill2003InvitationQT}.

Our derivation has four steps:
\begin{enumerate}
\item bound the Hilbert--Schmidt norm of a single-mode snapshot
  (Sec.~\ref{app:gg_onemode});
\item tensorize to the bipartite snapshot that the protocol actually
  measures (Sec.~\ref{app:gg_bipartite});
\item feed this pointwise bound into the exact Hoeffding variance
  formulas of Sec.~\ref{app:variance} to obtain state-independent
  variances for the moment estimators and both witnesses
  (Sec.~\ref{app:gg_variance});
\item apply Chebyshev's inequality to convert the variances into the
  sample-complexity bound (Secs.~\ref{app:gg_complexity}
  and~\ref{app:nonasymptotic_constants}).
\end{enumerate}

Throughout, $d = N + 1$ denotes the local Hilbert-space dimension at
Fock cutoff~$N$.  All bounds in this section are state-independent and
hold for any bipartite state with finite mean photon number
$\bar{n} < \infty$.

\subsection{One-mode snapshot norm bound}
\label{app:gg_onemode}

Fix a single mode~$M$ with Fock cutoff~$N$.  A homodyne measurement
with outcome~$x$ at phase~$\theta$ produces the shadow
matrix~$\rho_M(x,\theta) = \sum_{n,m=0}^N F_{nm}(x,\theta)\,
\ket{n}\!\bra{m}$~\eqref{eq:shadow_matrix}, whose
Hilbert--Schmidt norm squared is
\begin{equation}\label{eq:Delta_HS}
  \bigl\|\rho_M(x,\theta)\bigr\|_2^2
  = \sum_{n,m=0}^{N} |F_{nm}(x,\theta)|^2
  = \sum_{n,m=0}^{N} f_{nm}(x)^2 .
\end{equation}
The phase enters $F_{nm}(x,\theta)$ only through a unit-modulus factor,
so $|F_{nm}(x,\theta)| = |f_{nm}(x)|$ and the norm depends on the
outcome~$x$ alone.  It remains to bound the right-hand side
of~\eqref{eq:Delta_HS} uniformly in~$x$, which is exactly what the
pattern-function supremum estimate provides.

\begin{lemma}[Gill--Guta supremum bound~\cite{Gill2003InvitationQT}]
\label{lem:gill_guta}
The pattern-function amplitudes satisfy
\begin{equation}\label{eq:GG_bound}
  \sum_{0 \leq m \leq n \leq N}
  \|f_{nm}\|_\infty^2
  \;\leq\;
  C_{\mathrm{pf}}\,(N+1)^{7/3},
\end{equation}
where $C_{\mathrm{pf}} > 0$ is a universal constant
independent of~$N$.
\end{lemma}

This is Lemma~3.1 of Ref.~\cite{Gill2003InvitationQT}.  Its proof
controls the pattern functions through Airy-function asymptotics in the
transition region of the Hermite functions $\psi_n$ and their irregular
partners $\phi_n$, near the classical turning point
$|x| \approx \sqrt{2n+1}$ where the amplitudes are largest.  The
exponent~$7/3$ reflects the cubic Airy scaling of the supremum: the
diagonal amplitudes grow only slowly, $\|f_{nn}\|_\infty = O(n^{1/6})$,
and the double sum is dominated by a narrow band of near-diagonal terms
with $|n - m| \lesssim n^{1/3}$.

\begin{corollary}[One-mode HS norm]\label{cor:onemode}
For any measurement outcome~$(x,\theta)$,
\begin{equation}\label{eq:onemode_bound}
  \bigl\|\rho_M(x,\theta)\bigr\|_2^2
  \;\leq\;
  2\,C_{\mathrm{pf}}\,(N+1)^{7/3}.
\end{equation}
\end{corollary}

\begin{proof}
Extending the restricted sum in~\eqref{eq:GG_bound} to all
pairs $(n,m) \in \{0,\ldots,N\}^2$ using the symmetry
$f_{nm} = f_{mn}$:
\begin{equation}
  \sum_{n,m=0}^{N} \|f_{nm}\|_\infty^2
  = 2 \!\!\sum_{0 \leq m \leq n \leq N}\!\!
  \|f_{nm}\|_\infty^2
  - \sum_{n=0}^{N}\|f_{nn}\|_\infty^2
  \leq 2\,C_{\mathrm{pf}}\,(N+1)^{7/3}.
\end{equation}
Since $f_{nm}(x)^2 \leq \|f_{nm}\|_\infty^2$ pointwise,
the bound~\eqref{eq:onemode_bound} follows
from~\eqref{eq:Delta_HS}.
\end{proof}

\subsection{Bipartite tensorization}
\label{app:gg_bipartite}

The one-mode bound carries over to the bipartite case at once, because
the snapshot the protocol actually measures is a tensor product of two
single-mode shadows.  The partial-transposed bipartite shadow from a
single measurement round is
$\hat{\Sigma}_N^{T_B}(X)
= \rho_A(x_A,\theta_A) \otimes (\rho_B(x_B,\theta_B))^T$.
As in Sec.~\ref{app:variance}, we write
$\mathbf{z}(X) = \mathrm{vec}(\hat{\Sigma}_N^{T_B}(X))$ for the
vectorized snapshot.

\begin{proposition}[Bipartite snapshot norm]\label{prop:bipartite}
For any measurement outcome $X = (x_A, \theta_A, x_B, \theta_B)$,
\begin{equation}\label{eq:bipartite_HS}
  \bigl\|\hat{\Sigma}_N^{T_B}(X)\bigr\|_2^2
  = \bigl\|\rho_A(X_A)\bigr\|_2^2
    \cdot \bigl\|\rho_B(X_B)\bigr\|_2^2
  \;\leq\;
  4\,C_{\mathrm{pf}}^2\,(N+1)^{14/3}.
\end{equation}
\end{proposition}

\begin{proof}
The Hilbert--Schmidt norm of a tensor product factorizes:
$\|A \otimes B\|_2 = \|A\|_2 \cdot \|B\|_2$.
Since transposition preserves the HS norm,
$\|(\rho_B)^T\|_2 = \|\rho_B\|_2$.  Applying
Corollary~\ref{cor:onemode} to each mode yields the bound,
with the exponents adding: $7/3 + 7/3 = 14/3$.
\end{proof}

Proposition~\ref{prop:bipartite} bounds each snapshot separately.  The
variance formulas, however, depend on the snapshots only through their
\emph{average} squared norm, and the pointwise bound controls that
average just as well.  We name this average and record its bound.

\begin{definition}[Second-moment norm parameter]
\label{def:mu_N}
Define
\begin{equation}\label{eq:mu_def}
  \mu_N(\rho)
  \;:=\;
  \mathbb{E}_X\!\left[
  \bigl\|\hat{\Sigma}_N^{T_B}(X)\bigr\|_2^2
  \right]
  = \mathrm{Tr}[\mathbf{C}],
\end{equation}
where $\mathbf{C} = \mathbb{E}[\mathbf{z}\mathbf{z}^\dagger]$
is the second-moment matrix of the vectorized shadow
(cf.\ Eq.~\eqref{eq:sigma_bilinear}).
\end{definition}

A pointwise bound is automatically a bound on the average: since
$\|\hat{\Sigma}_N^{T_B}(X)\|_2^2
\leq 4\,C_{\mathrm{pf}}^2\,(N+1)^{14/3}$ for every outcome by
Proposition~\ref{prop:bipartite}, taking the expectation gives
\begin{equation}\label{eq:mu_bound}
  \mu_N(\rho)
  \;\leq\;
  4\,C_{\mathrm{pf}}^2\,(N+1)^{14/3}
\end{equation}
for all states~$\rho$ and all cutoffs~$N$.  This is the only property
of the snapshots used in the variance bounds below.

\subsection{Variance bounds for
  \texorpdfstring{$\hat{p}_2$}{p2} and
  \texorpdfstring{$\hat{p}_3$}{p3}}
\label{app:gg_variance}

We now insert the pointwise norm bound into the exact U-statistic
variance formulas from the Hoeffding decomposition
(Sec.~\ref{app:variance}).  We first note that the order-$k$ kernels
themselves obey pointwise bounds inherited directly from the snapshot
Hilbert--Schmidt norm.

\begin{lemma}[Kernel bounds]\label{lem:kernel_bounds}
For any measurement outcomes $X_1, \ldots, X_k$:
\begin{align}
  |h_2(X_1, X_2)|
  &\leq \bigl\|\hat{\Sigma}_N^{T_B}(X_1)\bigr\|_2
    \cdot \bigl\|\hat{\Sigma}_N^{T_B}(X_2)\bigr\|_2
  \leq 4\,C_{\mathrm{pf}}^2\,(N\!+\!1)^{14/3},
  \label{eq:h2_sup}\\[4pt]
  |h_3(X_1, X_2, X_3)|
  &\leq \prod_{i=1}^{3}
    \bigl\|\hat{\Sigma}_N^{T_B}(X_i)\bigr\|_2
  \leq 8\,C_{\mathrm{pf}}^3\,(N\!+\!1)^{7}.
  \label{eq:h3_sup}
\end{align}
\end{lemma}

\begin{proof}
For~$h_2$: since
$h_2(X_1,X_2)
= \mathrm{Tr}[\hat{\Sigma}_N^{T_B}(X_1)\,
  \hat{\Sigma}_N^{T_B}(X_2)]$,
the Cauchy--Schwarz inequality for the Hilbert--Schmidt
inner product gives
$|h_2| \leq \|\hat{\Sigma}_1\|_2\,\|\hat{\Sigma}_2\|_2$.
Each factor is bounded by
$\sqrt{4\,C_{\mathrm{pf}}^2\,(N+1)^{14/3}}
= 2\,C_{\mathrm{pf}}\,(N+1)^{7/3}$.

For~$h_3$: from the factorized
form~\eqref{eq:h3_fock_sm},
$|h_3| = |\mathrm{Tr}[A_1 A_2 A_3]|\,
|\mathrm{Tr}[B_1^T B_2^T B_3^T]|$.
By submultiplicativity of the trace,
$|\mathrm{Tr}[A_1 A_2 A_3]|
\leq \|A_1\|_2\,\|A_2\|_2\,\|A_3\|_2$
(H\"older with exponents $(2,4,4)$ or repeated
Cauchy--Schwarz), and similarly for the $B$-trace.
Hence $|h_3| \leq \prod_i \|A_i\|_2\,\|B_i\|_2
= \prod_i \|\hat{\Sigma}_i\|_2$.
\end{proof}

We now turn to the variances themselves, beginning with the second
moment.  From the Hoeffding decomposition~\eqref{eq:hoeffding_p2},
the leading variance component satisfies
\begin{equation}\label{eq:sigma21_GG}
  \sigma_{2,1}^2
  = \mathrm{Var}(G_2(X))
  \leq \mathbb{E}[G_2(X)^2]
  = \mathbb{E}\!\left[
    \bigl(\mathrm{Tr}[\hat{\Sigma}_N^{T_B}(X)\,\rho^{T_B}]
    \bigr)^2\right] - p_2^2.
\end{equation}
By Cauchy--Schwarz in the trace:
$|\mathrm{Tr}[\hat{\Sigma}\,\rho^{T_B}]|^2
\leq \|\hat{\Sigma}\|_2^2\,\|\rho^{T_B}\|_2^2
= \|\hat{\Sigma}\|_2^2\,p_2$.
Therefore
\begin{equation}\label{eq:sigma21_bound_GG}
  \sigma_{2,1}^2
  \;\leq\;
  \mu_N(\rho)\,p_2
  \;\leq\;
  4\,C_{\mathrm{pf}}^2\,(N+1)^{14/3}\,p_2.
\end{equation}
Substituting into~\eqref{eq:hoeffding_p2}:
\begin{equation}\label{eq:var_p2_GG}
  \mathrm{Var}(\hat{p}_2)
  \;\leq\;
  \frac{16\,C_{\mathrm{pf}}^2\,(N+1)^{14/3}\,p_2}{T}
  + O(T^{-2}).
\end{equation}
The factor $p_2 = \mathrm{Tr}[(\rho^{T_B})^2]$ keeps the bound
state-dependent, but never exceeds unity, so dropping it leaves the
unconditional bound
$\mathrm{Var}(\hat{p}_2)
\leq 16\,C_{\mathrm{pf}}^2\,(N+1)^{14/3}/T$, valid for every state.

The same Cauchy--Schwarz step applied to the order-3 first Hoeffding
projection~\eqref{eq:G3} gives
\begin{equation}
  |G_3(X)|
  = \bigl|\mathrm{Tr}[\hat{\Sigma}_N^{T_B}(X)\,
  (\rho^{T_B})^2]\bigr|
  \leq \|\hat{\Sigma}(X)\|_2\,\|(\rho^{T_B})^2\|_2
  = \|\hat{\Sigma}(X)\|_2\,\sqrt{p_4},
\end{equation}
where $p_4 = \mathrm{Tr}[(\rho^{T_B})^4] \leq p_2^2 \leq 1$.
Hence
\begin{equation}\label{eq:sigma31_bound_GG}
  \sigma_{3,1}^2
  \;\leq\;
  \mu_N(\rho)\,p_4
  \;\leq\;
  4\,C_{\mathrm{pf}}^2\,(N+1)^{14/3},
\end{equation}
and from~\eqref{eq:hoeffding_p3}:
\begin{equation}\label{eq:var_p3_GG}
  \mathrm{Var}(\hat{p}_3)
  \;\leq\;
  \frac{36\,C_{\mathrm{pf}}^2\,(N+1)^{14/3}}{T}
  + O(T^{-2}).
\end{equation}

The two witnesses inherit these bounds.  Each is a smooth function of
$\hat{p}_2$ and $\hat{p}_3$, so to leading order in $1/T$ its variance
is fixed by the influence-function (delta-method) expansion, with the
influence functions recorded in Sec.~\ref{app:jackknife_se}.  Take the
quadratic witness first.  The estimator
$\hat{W}_{\text{quad}} = \hat{p}_3 - \hat{p}_2^2$ has asymptotic
variance
\begin{equation}\label{eq:Veff_def}
  \mathrm{Var}(\hat{W}_{\text{quad}})
  = \frac{V_{\mathrm{eff}}}{T} + O(T^{-2}),
  \qquad
  V_{\mathrm{eff}}
  = 9\,\sigma_{3,1}^2 + 16\,p_2^2\,\sigma_{2,1}^2
  - 24\,p_2\,\sigma_{23,1},
\end{equation}
where $\sigma_{23,1} = \mathrm{Cov}(G_2,G_3)$.
Since $|\sigma_{23,1}|
\leq \sqrt{\sigma_{2,1}^2\,\sigma_{3,1}^2}$,
the triangle inequality gives
\begin{equation}
  V_{\mathrm{eff}}
  \;\leq\;
  9\,\sigma_{3,1}^2 + 16\,p_2^2\,\sigma_{2,1}^2
  + 24\,p_2\,\sqrt{\sigma_{2,1}^2\,\sigma_{3,1}^2}.
\end{equation}
Using $\sigma_{2,1}^2 \leq \mu_N p_2$ and
$\sigma_{3,1}^2 \leq \mu_N$ (dropping the $p_4 \leq 1$
factor), and $p_2 \leq 1$:
\begin{equation}
  V_{\mathrm{eff}}
  \;\leq\;
  \mu_N\bigl(9 + 16 + 24\bigr)
  = 49\,\mu_N.
\end{equation}
Thus
\begin{equation}\label{eq:var_W_GG}
  \mathrm{Var}(\hat{W}_{\text{quad}})
  \;\leq\;
  \frac{196\,C_{\mathrm{pf}}^2\,(N+1)^{14/3}}{T}
  + O(T^{-2}).
\end{equation}

The linear witness is simpler.  Because
$\hat{W}_{\text{lin}} = \hat{p}_3 - (3\hat{p}_2 - 1)/2$ is an exact
linear combination of the two U-statistics, no delta-method
approximation is involved: its influence function is exactly
$\psi_{\text{lin}} = 3G_3 - 3G_2$, giving the asymptotic variance
\begin{equation}\label{eq:Veff_lin_def}
  \mathrm{Var}(\hat{W}_{\text{lin}})
  = \frac{V_{\mathrm{eff}}^{\text{lin}}}{T} + O(T^{-2}),
  \qquad
  V_{\mathrm{eff}}^{\text{lin}}
  = 9\,\sigma_{3,1}^2 + 9\,\sigma_{2,1}^2
  - 18\,\sigma_{23,1}.
\end{equation}
The same bounds $\sigma_{2,1}^2 \leq \mu_N p_2 \leq \mu_N$,
$\sigma_{3,1}^2 \leq \mu_N$, and
$|\sigma_{23,1}| \leq \mu_N$ give
$V_{\mathrm{eff}}^{\text{lin}} \leq 36\,\mu_N$ and hence
\begin{equation}\label{eq:var_Wlin_GG}
  \mathrm{Var}(\hat{W}_{\text{lin}})
  \;\leq\;
  \frac{144\,C_{\mathrm{pf}}^2\,(N+1)^{14/3}}{T}
  + O(T^{-2}).
\end{equation}
Unlike the quadratic witness, $\hat{W}_{\text{lin}}$ is exactly
unbiased (Sec.~\ref{app:bias_correction}).

\textbf{Remark.}
The numerical prefactors ($16$, $36$, $144$, $196$) in the
bounds above are not optimized.  They arise from
keeping track of the Hoeffding combinatorial
coefficients ($4$, $9$) and the worst-case bound on
the covariance term.  Tighter prefactors can be obtained
by retaining the state-dependent quantities $p_2$ and
$p_4$ instead of bounding them by unity.

\subsection{Sample complexity}
\label{app:gg_complexity}

\begin{theorem}[State-independent sample
  complexity]\label{thm:GG_complexity}
For any bipartite state~$\rho$ with $\bar{n} < \infty$
and Fock cutoff~$N$, the U-statistic estimators
$\hat{p}_2$, $\hat{p}_3$, and the witnesses~$\hat{W}_{\text{lin}}, \hat{W}_{\text{quad}}$ can
be estimated to additive accuracy~$\varepsilon$ with
probability $\geq 1 - \delta$ using
\begin{equation}\label{eq:T_GG_cheb}
  T
  = O\!\left(\frac{(N+1)^{14/3}}{\delta\,\varepsilon^2}\right)
\end{equation}
samples, via Chebyshev's inequality.
\end{theorem}

\begin{proof}
By Chebyshev's inequality and~\eqref{eq:var_p2_GG}:
$\Pr(|\hat{p}_2 - p_2| > \varepsilon)
\leq \mathrm{Var}(\hat{p}_2)/\varepsilon^2
\leq 16\,C_{\mathrm{pf}}^2\,(N+1)^{14/3}/
(T\,\varepsilon^2)$.
Setting this equal to~$\delta$ and solving for~$T$
gives~\eqref{eq:T_GG_cheb}.
The same argument applies to~$\hat{p}_3$ and~$\hat{W}$,
with different numerical prefactors absorbed into the
$O(\cdot)$ notation.
\end{proof}

\subsection{Non-asymptotic constants and certification rules}
\label{app:nonasymptotic_constants}

The variance bounds derived so far carry $O(T^{-2})$ remainders from
the higher Hoeffding components, hidden inside the $O(\cdot)$ notation.
This subsection discharges those remainders rigorously and pins down
the explicit constants $K_2 = 32$ and $K_3 = 108$, together with the
certification constants $C_W$, exactly as quoted in the main text. It is
convenient to abbreviate
$\bar{\mu} := 4\,C_{\mathrm{pf}}^2\,(N+1)^{14/3}$, so that
$\|\hat{\Sigma}_N^{T_B}(X)\|_2^2 \leq \bar{\mu}$ holds pointwise
(Proposition~\ref{prop:bipartite}) and therefore
$\mu_N(\rho) \leq \bar{\mu}$.

\begin{lemma}[All Hoeffding components]\label{lem:all_components}
For $k \in \{2, 3\}$ and $1 \leq c \leq k$,
\begin{equation}\label{eq:sigma_kc_bound}
  \sigma_{k,c}^2 \;\leq\; \bar{\mu}^{\,c}\,p_{2(k-c)},
\end{equation}
with the convention $p_2 \leq 1$, $p_0 := 1$ (so in particular
$\sigma_{k,c}^2 \leq \bar{\mu}^{\,c}$).
\end{lemma}

\begin{proof}
Conditioning the kernel on $c$ of its arguments and using the
unbiasedness of the snapshots,
$\mathbb{E}[h_k \mid X_1, \ldots, X_c]
= \mathrm{Tr}\bigl[\hat{\Sigma}_N^{T_B}(X_1)\cdots
\hat{\Sigma}_N^{T_B}(X_c)\,(\rho_N^{T_B})^{k-c}\bigr]$.
For $c < k$, H\"older's inequality for Schatten norms (equivalently,
repeated Cauchy--Schwarz for the Hilbert--Schmidt inner product, keeping
one operator-norm factor) gives
$|\mathrm{Tr}[\hat{\Sigma}_1\cdots\hat{\Sigma}_c\,
(\rho_N^{T_B})^{k-c}]|
\leq \prod_{i\leq c}\|\hat{\Sigma}_i\|_2
\cdot \|(\rho_N^{T_B})^{k-c}\|_2
\leq \bar{\mu}^{c/2}\sqrt{p_{2(k-c)}}$
for $c = k-1, k-2$ (using
$\|\rho_N^{T_B}\|_{\mathrm{op}} \leq 1$,
Proposition~\ref{prop:opnorm}, where one operator-norm factor is
needed); for $c = k$ the kernel bounds of
Lemma~\ref{lem:kernel_bounds} give $|h_k| \leq \bar{\mu}^{k/2}$
pointwise.  Taking variances and expectations
yields~\eqref{eq:sigma_kc_bound}.
\end{proof}

\begin{theorem}[Explicit state-independent sample complexity]
\label{thm:explicit_constants}
Let $\rho$ be any bipartite state with $\bar{n} < \infty$, let
$N \geq 1$, and assume $T \geq 2\bar{\mu} + 1$.  Then
\begin{equation}\label{eq:var_explicit}
  \mathrm{Var}(\hat{p}_2)
  \;\leq\; \frac{8\,\bar{\mu}}{T}
  = \frac{32\,C_{\mathrm{pf}}^2 (N+1)^{14/3}}{T},
  \qquad
  \mathrm{Var}(\hat{p}_3)
  \;\leq\; \frac{27\,\bar{\mu}}{T}
  = \frac{108\,C_{\mathrm{pf}}^2 (N+1)^{14/3}}{T}.
\end{equation}
Consequently, for any $\varepsilon \in (0, 1]$ and
$\delta \in (0, 1)$, Chebyshev's inequality gives
$\Pr[|\hat{p}_k - p_k| > \varepsilon] \leq \delta$ whenever
\begin{equation}\label{eq:T_explicit}
  T \;\geq\;
  \frac{K_k\,C_{\mathrm{pf}}^2\,(N+1)^{14/3}}
       {\delta\,\varepsilon^2},
  \qquad K_2 = 32, \quad K_3 = 108,
\end{equation}
which is the bound quoted in the main text; these sample sizes
satisfy $T \geq 2\bar{\mu} + 1$ automatically (for
$\bar{\mu} \geq 1$).
\end{theorem}

\begin{proof}
The Hoeffding decomposition bounds the variance by
$\mathrm{Var}(\hat{p}_k) \leq \sum_{c=1}^{k}
\binom{k}{c}^2 \sigma_{k,c}^2/\binom{T}{c}$
[Eq.~\eqref{eq:hoeffding_general}].
For $k = 2$, Lemma~\ref{lem:all_components} gives
\begin{equation}
  \mathrm{Var}(\hat{p}_2)
  \leq \frac{4\bar{\mu}p_2}{T}
  + \frac{2\bar{\mu}^2}{T(T-1)}
  = \frac{4\bar{\mu}}{T}
    \Bigl[p_2 + \frac{\bar{\mu}}{2(T-1)}\Bigr]
  \leq \frac{8\bar{\mu}}{T},
\end{equation}
using $p_2 \leq 1$ and $\bar{\mu}/(2(T-1)) \leq 1$ for
$T \geq \bar{\mu}/2 + 1$.  For $k = 3$,
\begin{equation}
  \mathrm{Var}(\hat{p}_3)
  \leq \frac{9\bar{\mu}p_4}{T}
  + \frac{18\bar{\mu}^2}{T(T-1)}
  + \frac{6\bar{\mu}^3}{T(T-1)(T-2)}
  = \frac{9\bar{\mu}}{T}
    \Bigl[p_4 + \frac{2\bar{\mu}}{T-1}
    + \frac{2\bar{\mu}^2}{3(T-1)(T-2)}\Bigr].
\end{equation}
For $T \geq 2\bar{\mu} + 1$ (and $\bar{\mu} \geq 1$):
$2\bar{\mu}/(T-1) \leq 1$ and
$\tfrac{2}{3}\bar{\mu}^2/((T-1)(T-2)) \leq
\tfrac{2}{3}\cdot\tfrac{1}{2}\cdot
\bar{\mu}/(T-2) \leq \tfrac{1}{3}$, so the bracket is at most
$1 + 1 + \tfrac{1}{3} < 3$ and
$\mathrm{Var}(\hat{p}_3) \leq 27\bar{\mu}/T$.
Finally, $K_k = k^2 \cdot 4 \cdot
(\text{bracket bound}) = \{8, 27\} \times 4 = \{32, 108\}$
in units of $C_{\mathrm{pf}}^2(N+1)^{14/3}$.
\end{proof}

\subsubsection{Certification rule for the linear witness}

The certification rule of the main text follows by propagating the
per-moment bound~\eqref{eq:T_explicit} through the witness.  Since
$\hat{W}_{\text{lin}}$ is linear in the two moments, the triangle
inequality controls its deviation directly,
\begin{equation}\label{eq:Wlin_triangle}
  |\hat{W}_{\text{lin}} - W_{\text{lin}}|
  \;\leq\; \tfrac{3}{2}\,|\hat{p}_2 - p_2| + |\hat{p}_3 - p_3| .
\end{equation}
Requiring each moment to be accurate at confidence $1 - \delta/2$ and
joining the two events by a union bound then yields the rule: the
observation $\hat{W}_{\text{lin}} + \varepsilon < 0$ certifies
$W_{\text{lin}} < 0$ with probability $\geq 1 - \delta$ whenever
$T \geq C_W\,C_{\mathrm{pf}}^2 (N+1)^{14/3}/(\delta\varepsilon^2)$.
Splitting the target accuracy as
$\tfrac{3}{2}\varepsilon_2 + \varepsilon_3 = \varepsilon$ and
optimizing over the two per-moment conditions gives
\begin{equation}\label{eq:CW_lin}
  C_W = \Bigl(\tfrac{3}{2}\sqrt{2K_2}
  + \sqrt{2K_3}\Bigr)^2 \approx 713,
\end{equation}
as quoted in the main text.  No bias correction enters here, because
$\hat{W}_{\text{lin}}$ is exactly unbiased.

\subsubsection{Certification rule for the quadratic witness}
\label{app:quad_rule}

The quadratic witness needs one extra step, because the
$\hat{p}_2^2$ term is nonlinear in $\hat{p}_2$.  Its error propagation
therefore carries the additional factor $(p_2 + \hat{p}_2)$ announced
in the main text,
\begin{equation}\label{eq:Wquad_triangle}
  |\hat{W}_{\text{quad}} - W_{\text{quad}}|
  \;\leq\; |\hat{p}_3 - p_3|
  + (p_2 + \hat{p}_2)\,|\hat{p}_2 - p_2|,
\end{equation}
which follows from the factorization
$\hat{p}_2^2 - p_2^2 = (\hat{p}_2 + p_2)(\hat{p}_2 - p_2)$.  On the
event $|\hat{p}_2 - p_2| \leq \varepsilon_2$ the prefactor is bounded,
$p_2 + \hat{p}_2 \leq 2p_2 + \varepsilon_2 \leq 2 + \varepsilon_2$, so
the same union-bound construction (confidence $\delta/2$ per moment,
error split
$(2+\varepsilon_2)\varepsilon_2 + \varepsilon_3 = \varepsilon$) yields
the following rule.

\begin{proposition}[Rigorous rule, quadratic witness]
\label{prop:quad_rule}
For $\varepsilon \in (0,1]$, $\delta \in (0,1)$, the implication
\begin{equation}
  \hat{W}_{\text{quad}} + \varepsilon < 0
  \quad\Longrightarrow\quad
  W_{\text{quad}} < 0
  \quad\text{with probability} \geq 1 - \delta
\end{equation}
holds whenever
\begin{equation}\label{eq:T_quad_rule}
  T \;\geq\;
  \frac{C_W^{\mathrm{quad}}\,
        C_{\mathrm{pf}}^2\,(N+1)^{14/3}}
       {\delta\,\varepsilon^2},
  \qquad
  C_W^{\mathrm{quad}}
  = \Bigl(3\sqrt{2K_2} + \sqrt{2K_3}\Bigr)^2
  \approx 1498,
\end{equation}
where the conservative factor $3 \geq 2 + \varepsilon_2$ bounds
the moment sum; for small target accuracies the effective constant
approaches $\bigl(2\sqrt{2K_2} + \sqrt{2K_3}\bigr)^2 \approx 942$.
Note that this rule uses the \emph{raw} estimator
$\hat{W}_{\text{quad}}$; no bias correction is needed because the
propagation bound~\eqref{eq:Wquad_triangle} controls the deviation
directly.
\end{proposition}

\textbf{Remark (Asymmetric cutoffs).}  When the two modes carry
different photon numbers, the snapshot norm factorizes per mode
[Eq.~\eqref{eq:bipartite_HS}], so independent cutoffs $N_A \neq N_B$
give $T = O\bigl((N_A+1)^{7/3}(N_B+1)^{7/3}/\varepsilon^2\bigr)$,
letting one truncate the less-excited mode more aggressively to reduce
the total sample count.

\subsection{Discussion}
\label{app:gg_discussion}

The exponent $7/3$ in Lemma~\ref{lem:gill_guta} is attained near the
Airy transition region, where the pattern functions are maximally
oscillatory.  Within the present argument it cannot be improved,
because the bound treats every measurement outcome in the worst case,
through the $\ell^\infty$ norm of the pattern functions, and so
discards all information about how the outcomes are actually
distributed.  Any further sharpening must therefore use that discarded
structure, in particular the quadrature distribution $\mu(x)$.  This is
exactly the route taken by the state-dependent geometric-decay analysis
of Sec.~\ref{app:gd_bounds}, which recovers an $N$-independent rate at
the cost of restricting the class of states.


\section{Negativity bounds from truncated moments}
\label{app:trunc_neg}

The negativity bound of the main text applies to the
exact normalized moments $(p_1, p_2, p_3)$ with $p_1 = 1$.  In the
truncated protocol, the homodyne estimators target the projected
moments
$p_k^{(N)} = \mathrm{Tr}[(\rho_N^{T_B})^k]$
of the subnormalized operator
$\rho_N = P_N\,\rho\,P_N$, where
$P_N = \Pi_N \otimes \Pi_N$ is the local Fock projector\ (notation as in the main text).
Since $t_N := \mathrm{Tr}[\rho_N] \leq 1$ in general, one cannot
simply substitute $p_k^{(N)}$ into the normalized bound.
Here we prove dimension-free bounds that apply directly to the
subnormalized truncated moments and transfer rigorously to the
full-state negativity.

\subsection{Negativity under local compression}
\label{app:neg_compression}

For any positive trace-class $\tau$ (not necessarily normalized),
define its \emph{unnormalized negativity}
\begin{equation}\label{eq:unnorm_neg}
  \mathcal{N}_*(\tau) :=
  \mathrm{Tr}\bigl[(\tau^{T_B})_-\bigr]
  = \frac{\|\tau^{T_B}\|_1 - \mathrm{Tr}[\tau]}{2}.
\end{equation}
When $\mathrm{Tr}[\tau] = 1$, this is the usual negativity.  Physically,
projecting onto a finite Fock window discards weight from $\rho^{T_B}$
and can only shrink the negative part, so the truncated negativity never
overshoots the true one:

\begin{lemma}[Compression monotonicity]\label{lem:compression}
$\mathcal{N}_*(\rho_N) \leq \mathcal{N}(\rho)$.
\end{lemma}

\begin{proof}
We use the variational expression for the negative-eigenvalue weight,
$\mathrm{Tr}[X_-] = \max_{0 \leq M \leq I}(-\mathrm{Tr}[MX])$ for
self-adjoint trace-class~$X$ (the optimal $M$ projects onto the negative
eigenspace).  Restricting the support to the Fock window,
\begin{equation}
  \mathrm{Tr}\bigl[(P_N A P_N)_-\bigr]
  = \max_{0 \leq M \leq P_N}\bigl(-\mathrm{Tr}[MA]\bigr),
\end{equation}
where $A = \rho^{T_B}$.  Writing $A = A_+ - A_-$ with
$A_\pm \geq 0$, every $0 \leq M \leq P_N \leq I$ satisfies
$-\mathrm{Tr}[MA] = \mathrm{Tr}[MA_-] - \mathrm{Tr}[MA_+]
\leq \mathrm{Tr}[A_-] = \mathcal{N}(\rho)$.
\end{proof}

Any lower bound on $\mathcal{N}_*(\rho_N)$ is therefore
automatically a lower bound on $\mathcal{N}(\rho)$.

\subsection{Spectral range for subnormalized partial transposes}
\label{app:spectral_range}

\begin{lemma}[Subnormalized spectral range]\label{lem:spectral_subnorm}
If $\tau \geq 0$ with $\mathrm{Tr}[\tau] = t$, then every
eigenvalue of $\tau^{T_B}$ lies in $[-t/2,\; t]$.
In particular,
$\mathrm{Tr}[(\tau^{T_B})^3] \geq -(t/2)\,\mathrm{Tr}[(\tau^{T_B})^2]$.
\end{lemma}

\begin{proof}
Define the normalized state $\sigma = \tau/t$.  By
Proposition~\ref{prop:opnorm},
$-\tfrac{1}{2}I \leq \sigma^{T_B} \leq I$.
Multiplying by~$t$ gives
$-\tfrac{t}{2}I \leq \tau^{T_B} \leq tI$.
The scalar inequality $x^3 \geq -(t/2)x^2$ on $[-t/2, t]$
and functional calculus yield
$\mathrm{Tr}[(\tau^{T_B})^3]
\geq -(t/2)\,\mathrm{Tr}[(\tau^{T_B})^2]$.
\end{proof}

\subsection{Truncated negativity bounds}
\label{app:trunc_neg_bounds}

Write $q_k = p_k^{(N)} = \mathrm{Tr}[(\rho_N^{T_B})^k]$ for
brevity, with $q_1 = t_N$.  Define the \emph{truncated
third-order deficit}
\begin{equation}\label{eq:trunc_deficit}
  \Delta_N := q_2^2 - t_N\, q_3
  = \bigl(p_2^{(N)}\bigr)^2 - t_N\, p_3^{(N)}.
\end{equation}
Note that this is the Hankel determinant
$\det\bigl(\begin{smallmatrix} t_N & q_2 \\
q_2 & q_3 \end{smallmatrix}\bigr)$
of the subnormalized moment sequence, not the deficit
$q_2^2 - q_3$ used by the normalized witness.

\begin{theorem}[Truncated negativity bounds]\label{thm:trunc_neg}
Let $\rho$ be a bipartite state, $P_N$ the local Fock
projector, $\rho_N = P_N\rho P_N$ with $t_N = \mathrm{Tr}[\rho_N] > 0$,
and $p_k^{(N)} = \mathrm{Tr}[(\rho_N^{T_B})^k]$.
If $\Delta_N > 0$, then:

\emph{(i) Cubic-root bound.}
\begin{equation}\label{eq:trunc_cubic}
  \mathcal{N}(\rho) \;\geq\; u_{N,*},
\end{equation}
where $u_{N,*}$ is the smallest positive real root of
\begin{equation}\label{eq:trunc_cubic_poly}
  t_N\,u^3 + 2\,p_2^{(N)}\,u^2 + p_3^{(N)}\,u
  + t_N\,p_3^{(N)} - \bigl(p_2^{(N)}\bigr)^2 = 0.
\end{equation}

\emph{(ii) Closed-form rational bound.}
\begin{equation}\label{eq:trunc_rational}
  \mathcal{N}(\rho) \;\geq\;
  \frac{\bigl(p_2^{(N)}\bigr)^2 - t_N\, p_3^{(N)}}
       {t_N\, p_2^{(N)} + p_3^{(N)} + t_N^3/4}\,.
\end{equation}

\emph{(iii) Simpler corollary.}
\begin{equation}\label{eq:trunc_simple}
  \mathcal{N}(\rho) \;\geq\;
  \frac{t_N\bigl[\bigl(p_2^{(N)}\bigr)^2 - t_N\, p_3^{(N)}\bigr]}
       {\bigl(t_N^2 + p_2^{(N)}\bigr)^2}\,.
\end{equation}

All bounds are dimension-free and reduce to those of
the full-state negativity theorem (the negativity-bound appendix of the main text) when $t_N = 1$.
\end{theorem}

\begin{proof}
By Lemma~\ref{lem:compression} it suffices to bound
$\mathcal{N}_*(\rho_N)$.  Write $B = \rho_N^{T_B}$ with positive
eigenvalues $\{\alpha_j\}$ and negative magnitudes $\{\beta_k > 0\}$, so
that $\sum_j \alpha_j = t_N + \nu$ with $\nu = \mathcal{N}_*(\rho_N)$,
and set $y = \sum_k \beta_k^2$.  The argument is the normalized
negativity proof of the main-text appendix, carried through with the
total weight $t_N$ in place of $p_1 = 1$.  Cauchy--Schwarz on the
positive eigenvalues gives $\sum_j \alpha_j^3 \geq (q_2 - y)^2/(t_N +
\nu)$, while the negative part admits two bounds: $\beta_k \leq \nu$
gives $\sum_k \beta_k^3 \leq \nu y$, and the subnormalized spectral
range $\beta_k \leq t_N/2$ (Lemma~\ref{lem:spectral_subnorm}) gives
$\sum_k \beta_k^3 \leq (t_N/2)\,y$.

For~(i), inserting $\sum_k \beta_k^3 \leq \nu y$ and running the same
two-case analysis in $\nu$ as in the appendix
($\nu^2 \lessgtr q_2$) yields $g_{t_N}(\nu) \geq 0$, with
$g_{t_N}(u) = t_N u^3 + 2 q_2 u^2 + q_3 u + t_N q_3 - q_2^2$ the cubic
of~\eqref{eq:trunc_cubic_poly}.  Since $g_{t_N}(0) = t_N q_3 - q_2^2
= -\Delta_N < 0$, its smallest positive root $u_{N,*}$ satisfies
$\nu \geq u_{N,*}$.  For~(ii), inserting instead the spectral bound
$\sum_k \beta_k^3 \leq (t_N/2)\,y$ and the matching case split (now on
$t_N\nu/2 \lessgtr q_2$, the subnormalized analogue of $\nu^2 \lessgtr
q_2$) gives $\nu\,(t_N q_2 + q_3 + t_N^3/4) \geq q_2^2 - t_N q_3$, which
is~\eqref{eq:trunc_rational}.  Part~(iii) then follows from
$\Delta_N > 0$, i.e.\ $q_3 < q_2^2/t_N$, which yields
$t_N q_2 + q_3 + t_N^3/4 < (t_N^2 + q_2)^2/t_N$.
\end{proof}

\textbf{Remark (Relation to the truncated quadratic witness).}
The condition $\Delta_N > 0$ is the Hankel determinant
condition for the \emph{subnormalized} moment sequence
$(t_N, q_2, q_3)$.  
It coincides with the violation of the
$t_N$-corrected truncated quadratic witness of the main text,
$\Delta_N = -W_{\text{quad}}^{(N)}
= \bigl(p_2^{(N)}\bigr)^2 - t_N\,p_3^{(N)}$.
The main text adopts the simplified ($t_N = 1$) witness
$\widehat{W}_{\text{quad}}^{(N)} = q_3 - q_2^2$ for detection.
Since $t_N \leq 1$, the condition $\Delta_N > 0$ is
strictly stronger than $\widehat{W}_{\text{quad}}^{(N)} < 0$: there exist
truncated states with $\widehat{W}_{\text{quad}}^{(N)} \geq 0$ but $\Delta_N > 0$.
Both quantities can be estimated from the same homodyne
data---$t_N$ is the truncated first moment, a degree-1
U-statistic with negligible variance overhead; see
the truncation-validity appendix of the main text for the
detection-power comparison of the two conventions.

\textbf{Remark (Reduction to full-state bounds).}
When $t_N = 1$ (the truncation captures the full state),
$\Delta_N = q_2^2 - q_3 = p_2^2 - p_3$,
$g_1(u) = u^3 + 2q_2 u^2 + q_3 u + q_3 - q_2^2$,
and Theorem~\ref{thm:trunc_neg} reduces to
the full-state negativity theorem.

\subsection{Pure-state specialization}
\label{app:trunc_neg_pure}

For a pure input state $\ket{\psi}$, the projected operator
$\rho_N = P_N\ket{\psi}\!\bra{\psi}P_N =
\ket{\tilde\psi}\!\bra{\tilde\psi}$ is a rank-one operator of norm
$\braket{\tilde\psi|\tilde\psi} = t_N$, so
$q_2 = \mathrm{Tr}[\rho_N^2] = t_N^2$ using the invariance of the
Hilbert--Schmidt norm under partial transposition.  Substituting
$q_2 = t_N^2$ into the cubic~\eqref{eq:trunc_cubic_poly} factors it
as
\begin{equation}\label{eq:trunc_cubic_pure}
  t_N u^3 + 2t_N^2 u^2 + q_3 u + t_N q_3 - t_N^4
  = (u + t_N)\bigl(t_N u^2 + t_N^2 u + q_3 - t_N^3\bigr),
\end{equation}
whose positive root gives the truncated pure-state bound quoted in
the main text:
\begin{equation}\label{eq:trunc_pure_bound}
  \mathcal{N}\bigl(\ket{\psi}\!\bra{\psi}\bigr)
  \;\geq\;
  \frac{-t_N^2 + \sqrt{5\,t_N^4 - 4\,t_N\,p_3^{(N)}}}{2\,t_N}
  \;=:\; \mathcal{N}_{\text{pure}}^{(N)},
\end{equation}
valid whenever $\Delta_N = t_N^4 - t_N p_3^{(N)} > 0$, i.e.\
$p_3^{(N)} < t_N^3$.  Setting $t_N = 1$ recovers the pure-state
bound $\mathcal{N}_{\text{pure}} = \bigl(-1 + \sqrt{5 - 4p_3}\bigr)/2$
of the main text.

\end{document}